\documentclass[a4paper,UKenglish]{lipics}
 
\usepackage{microtype}

\usepackage{amsmath}
\usepackage{amsthm}

\graphicspath{{./images/}}

\usepackage{cite}
\usepackage{graphicx}
\usepackage{balance}  
\usepackage[utf8]{inputenc}
\usepackage{hyperref}
\usepackage{import}
\usepackage{xspace}
\usepackage{amssymb}
\usepackage{multirow}
\usepackage{amsthm}
\usepackage{multirow}
\usepackage{textcomp}
\usepackage{booktabs}
\usepackage{enumerate}
\usepackage[normalem]{ulem} 
\usepackage{booktabs}
\usepackage{array}
\usepackage{upgreek}
\usepackage{svg}
\usepackage{paralist}
\usepackage{enumerate}
\usepackage{enumitem}
\setlist{noitemsep,topsep=1pt,parsep=1pt,partopsep=0pt,leftmargin=15pt}
\newlist{lenum}{enumerate}{2}
\setlist[lenum]{leftmargin=*, labelsep=0.5ex, itemsep=.5em} 

\usepackage{fancyhdr}
\usepackage{datetime}
\usepackage{marginnote}

\usepackage[toc,page]{appendix}
\usepackage{minitoc}

\newcommand{\myparagraph}{\textbf}

\renewcommand{\paragraph}{\subparagraph}
\makeatletter
\renewcommand\subparagraph{\@startsection{subparagraph}{5}{\z@}%
                                       {1.25ex \@plus1ex \@minus .2ex}%
                                       {-1em}%
                                      {\sffamily\normalsize\bfseries}}
\makeatother

\usepackage{amsmath}

\newcommand{\RationalsPos}{{\mathbb Q}^+}
\newcommand{\domQPos}{\dom_{\RationalsPos}}

\newcommand{\SigmaB}{\Sigma_{\B}}
\newcommand{\SigmaBIn}{\Sigma_{\B,\In}}

\makeatletter
\def\thm@space@setup{%
  \thm@preskip=1.4ex 
  \thm@postskip=1.4\thm@preskip 
}
\makeatother

\newcommand{\adom}{\textit{adom}}       
\newcommand{\adomExt}{\textit{adom}^*}        

\newcommand{\Inz}{\In^0}

\newcommand{\BPID}{\B=\tup{\P,\I,\D}}
\newcommand{\PID}{\tup{\P,\I,\D}}

\newcounter{myexample}[section]

\newcounter{firstexample}[section]

\renewcommand{\paragraph}{\subparagraph}
\makeatletter
\renewcommand\subparagraph{\@startsection{subparagraph}{5}{\z@}%
                                       {1.25ex \@plus1ex \@minus .2ex}%
                                       {-1em}%
                                      {\sffamily\normalsize\bfseries}}
\makeatother

\newcommand{\decd}{Dec.}

\newcommand{\tauB}{\tau_\B}

\makeatletter
\def\old@comma{,}
\catcode`\,=13
\def,{%
  \ifmmode%
    \old@comma\discretionary{}{}{}%
  \else%
    \old@comma%
  \fi%
}
\makeatother

\newcommand{\reg}{\textit{reg}}
\newcommand{\Breg}{\B_{\reg}}
\newcommand{\Preg}{\P_{\reg}}
\newcommand{\Ireg}{\I_{\reg}}
\newcommand{\Creg}{\C_{\reg}}
\newcommand{\Dreg}{\D_{\reg}}

\newcommand{\supremum}{\infty}

\newcommand{\observations}[1]{\subsubsection*{Observations:}}
\newcommand{\vocabulary}[1]{\subsubsection*{Vocabulary and Symbols:}}
\newcommand{\rules}[1]{\subsubsection*{Rules:}}
\newcommand{\problem}[1]{\subsubsection*{Problem:}}
\newcommand{\relations}[1]{\subsubsection*{Relations:}}
\newcommand{\programname}[1]{\subsubsection*{Program Name:}}

\makeatletter
\newsavebox\myboxA
\newsavebox\myboxB
\newlength\mylenA

\newcommand*\xoverline[2][0.75]{%
    \sbox{\myboxA}{$\m@th#2$}%
    \setbox\myboxB\null
    \ht\myboxB=\ht\myboxA%
    \dp\myboxB=\dp\myboxA%
    \wd\myboxB=#1\wd\myboxA
    \sbox\myboxB{$\m@th\overline{\copy\myboxB}$}
    \setlength\mylenA{\the\wd\myboxA}
    \addtolength\mylenA{-\the\wd\myboxB}%
    \ifdim\wd\myboxB<\wd\myboxA%
       \rlap{\hskip 0.5\mylenA\usebox\myboxB}{\usebox\myboxA}%
    \else
        \hskip -0.5\mylenA\rlap{\usebox\myboxA}{\hskip 0.5\mylenA\usebox\myboxB}%
    \fi}
\makeatother

\newcommand{\Dig}{\textit{Dig}}
\newcommand{\PiSucc}{\Pi^\textit{succ}}

\newcommand{\tepsilon}{\bar \epsilon}
\newcommand{\tT}{\xoverline{T}}

\newcommand{\tk}{\bar k}
\newcommand{\tK}{\xoverline{K}}

\newcommand{\tW}{\xoverline{W}}

\newcommand{\tU}{\xoverline{U}}

\newcommand{\tX}{\xoverline{X}}

\newcommand{\tY}{\xoverline{Y}}

\newcommand{\ttr}{\bar t}

\newcommand{\tomega}{\bar \omega}
\newcommand{\tb}{\bar b}

\newcommand{\In}{\textit{In}}
\newcommand{\tuple}[1]{\bar{#1}}
\newcommand{\tu}{\tuple{u}}

\newcommand{\Place}{\textit{Place}}
\newcommand{\State}{\textit{State}}
\newcommand{\Done}{\textit{Completed}}
\newcommand{\Succ}{\textit{Succ}}
\newcommand{\Path}{\textit{Path}}
\newcommand{\Reach}{\textit{Reach}}

\newcommand{\Key}{\textit{Step}}

\newcommand{\TGuess}{\textit{Trans}}
\newcommand{\NonTGuess}{\textit{NotTrans}}
\newcommand{\InGuess}{\textit{Moved}}
\newcommand{\NonInGuess}{\textit{NotMoved}}

\newcommand{\First}{\textit{First}}
\newcommand{\Last}{\textit{Last}}
\newcommand{\Const}{\textit{Const}}

\newcommand{\Digit}{\textit{Digit}}
\newcommand{\SCC}{\mathit{SCC}}

\newcommand{\Moved}{\InGuess}
\newcommand{\NotMoved}{\NonInGuess}
\newcommand{\Trans}{\TGuess}
\newcommand{\NotTrans}{\NonTGuess}
\newcommand{\Completed}{\Done}

\newcommand{\grn}{\textit{gnd}}
\newcommand{\mdl}{\models_{\textit{brave}}}
\newcommand{\Qtest}{Q_{\textit{test}}}

\newcommand{\done}{\textit{Done}}
\newcommand{\CurrentR}{\textit{Current}_R}
\newcommand{\fail}{\textit{fail}_{\textit{MM}}}
\newcommand{\faili}{\textit{fail}_i}

\newcommand{\studyplan}{\textit{StudyPlan}}

\newcommand{\conditional}{\textit{Conditional}}
\newcommand{\accepted}{\textit{Pre-enrolled}}
\newcommand{\registered}{\textit{Registered}}
\newcommand{\course}{\textit{course}}
\newcommand{\program}{\textit{program}}

\newcommand{\first}{1$^\text{st}$}

\newcommand{\nth}[1]{#1$^\text{th}$}

\newcommand{\rra}[1]
{\renewcommand{\arraystretch}{#1}}

\newcommand{\tx}{{\bar x}} 
\newcommand{\ts}{{\bar s}} 
\newcommand{\td}{\bar d}   
\newcommand{\tc}{\bar c}   

\newcommand{\instable}{\textit{Instable}}        

\newcommand{\PiPIpocl}{\Pi_{\P,\I}^{\textsl{po,cl}}} 
 
\newcommand{\PiQtest}{\Pi_Q^\textsl{test}}

\newcommand{\PiPQposfr}{\Pi_{\P,Q}^{\textsl{po,fr}}} 
\newcommand{\PiPposfr}{\Pi_{\P}^{\textsl{po,fr}}} 
 
\newcommand{\PiPIQacyccl}{\Pi_{\P,\I}^{\textsl{ac,cl}}} 
\newcommand{\Picl}{\Pi_{\P}^{\textsl{cl}}} 
 
\newcommand{\PiPIQtest}{\Pi_{\P,\I,Q}^{\textsl{test}}} 
\newcommand{\PiPIQtestrowo}{\Pi_{\P,Q}^{{\textsl{test,ro}}}}
\newcommand{\PiPQrowofr}{\Pi_{\P,Q}^{\textsl{ro,fr}}} 
\newcommand{\PiProwofr}{\Pi_{\P}^{\textsl{ro,fr}}} 
\newcommand{\PiPQrowocl}{\Pi_{\P,Q}^{\textsl{ro,cl}}} 
\newcommand{\PiProwocl}{\Pi_{\P}^{\textsl{ro,cl}}} 
\newcommand{\PiPQrowo}{\Pi_{\P,Q}^{\textsl{ro}}}

\newcommand{\PiPIQtestrowoOpen}{\Pi_{\P,Q}^{{\textsl{test,ro,op}}}}
\newcommand{\Exec}{\textit{Exec}}   

\newcommand{\blank}
     {{\rule{0.5em}{0.5pt}\hspace{0.05em}}\xspace} 
\newcommand{\blankk}
    {\raisebox{0.8ex}{-}\kern-0.4em b} 

\theoremstyle{plain}
\newtheorem{proposition}[theorem]{Proposition}

\newcommand{\LE}[1]{E_{\text{\emph{#1}}}}
\newcommand{\LW}[1]{W_{\text{\emph{#1}}}}
\newcommand{\tlbl}[1]{\text{\emph{#1}}}

\newcommand{\ei}{\emph{(i)~}}
\newcommand{\eii}{\emph{(ii)~}}
\newcommand{\eiii}{\emph{(iii)~}}
\newcommand{\eiv}{\emph{(iv)~}}
\newcommand{\ev}{\emph{(v)~}}

\newcommand{\msccs}{mscCS}
\newcommand{\msceco}{mscEco}
\newcommand{\emcl}{compLogic}

\newcommand{\econ}{econ}

\newcommand{\transf}{\text{\emph{intl}}}
\newcommand{\dateconditional}{\nth{30} Sep}

\newcommand{\ord}{\text{\emph{reg}}}

\newcommand{\tstamp}{\T}
\newcommand{\act}[1]{(`#1')}
\newcommand{\qact}[1]{`#1'}

\newcommand{\col}{\colon}

\newcommand{\quotes}[1]{\lq\lq#1\rq\rq}         

\newcommand{\dom}{\mathit{dom}}

\newcommand{\true}{\textit{true}}

\newcommand{\tpl}[1]{\bar{#1}}				
\newcommand{\eat}[1]{}

 \newcommand{\B}{\mathcal{B}}
\newcommand{\C}{\mathcal{C}} \newcommand{\D}{\mathcal{D}}

\newcommand{\I}{\mathcal{I}}

 \renewcommand{\P}{\mathcal{P}}
 
 \newcommand{\T}{\mathcal{T}}

\newcommand{\defterm}[1]{\mbox{\underline{\it\smash{#1}\vphantom{\lower.1ex\hbox
{x}}}}}

\newcommand{\la}{\leftarrow}

\newcommand{\set}[1]{\{#1\}}                      

\newcommand{\tup}[1]{\langle #1\rangle}            

\newcommand{\undec}{\textsc{Undec.}}

\newcommand{\ACz}{\textsc{AC}\ensuremath{^0}\xspace}

\newcommand{\PTIME}{\textsc{PTime}\xspace}
\newcommand{\PSPACE}{\textsc{PSpace}\xspace}
\newcommand{\EXPSPACE}{\textsc{ExpSpace}\xspace}
\newcommand{\EXPTIME}{\textsc{ExpTime}\xspace}
\newcommand{\NEXPTIME}{\textsc{NExpTime}\xspace}
\newcommand{\coNEXPTIME}{\textsc{co-NExpTime}\xspace}

\newcommand{\NP}{\textsc{NP}\xspace}
\newcommand{\coNP}{\textsc{co-NP}\xspace}
\newcommand{\PiPTWO}{\Pi^\textsc{P}_2\xspace}

\theoremstyle{definition} 

\newcommand{\terminate}{{\textit{end\/}}}
\newcommand{\start}{{\textit{start\/}}}

\newcommand{\DABP}
	{MDBP\xspace}    

\newcommand{\DATALOG}{\text{Datalog}\xspace}   
\newcommand{\DATALOGNEG}
	{\text{Datalog}^{\neg}\xspace}

\newcommand{\arity}{\xspace\textit{arity\/}\xspace}

\newcommand{\dummy}{\xspace\textit{dummy}\xspace}

\begin{document}
\doparttoc 
\faketableofcontents 




\EventShortName{}%
\title{Query Stability in Monotonic Data-Aware Business Processes
[Extended Version]%
\footnote{This report is the extended version of a paper accepted at the 19th International Conference on Database Theory (ICDT 2016),
March 15-18, 2016 - Bordeaux, France.
}}


\author{Ognjen Savkovi\'c}
\author{Elisa Marengo}
\author{Werner Nutt}
\affil{Free University of Bozen-Bolzano\\
       39100 Bolzano, Italy\\
\texttt{\{ognjen.savkovic,elisa.marengo,werner.nutt\}@unibz.it}}

\Copyright{Ognjen Savkovic and Elisa Marengo and Werner Nutt}

\subjclass{H.2.4 [Systems]: Relational databases}
\keywords{Business Processes, Query Stability}


\maketitle


\begin{abstract}
Organizations continuously accumulate data,
often according to some business processes.
If one poses a query over such data for decision support, 
it is important to know whether the query is \emph{stable,} that is,
whether the answers will stay the same or may change in the future
because business processes may add further data.
We investigate query stability for conjunctive queries.
To this end, we define a formalism that combines 
an explicit representation of the control flow of a process 
with a specification of how data is read and inserted into the database.
We consider different restrictions of the process model 
and the state of the system, such as
negation in conditions,
cyclic executions,
read access to written data,
presence of pending process instances, and
the possibility to start fresh process instances. 
We identify for which facet combinations
stability of conjunctive queries is decidable and provide encodings
into variants of Datalog
that are optimal with respect to the worst-case complexity of the problem.

\end{abstract}


\section{Introduction}

Data quality focuses on understanding how much data 
is fit for its intended use.
This problem has been investigated in database theory, 
considering aspects such as consistency, currency, and completeness%
~\cite{Wenfei-Consistency-VLDB-07, Fan:Et:Al-Currency-TODS,Razniewski-Completeness-VLDB-11}.
A question that these approaches 
consider only marginally is 
where data originates and how it evolves.

{Although} in general a database may evolve in arbitrary ways,
often data are generated according to some 
business process, implemented in an information system 
that accesses the DB.
We believe that analyzing how business processes generate 
data allows one to gather additional information on their fitness 
for use.
In this work, we focus on a particular aspect of data quality, that is
the problem whether a business process that reads from and writes into
a database can affect the answer of a query or whether the answer will
not change as a result of the process. We refer to this problem as
query stability.

For example, consider a student registration process at a university. 
The university maintains a relation $\textit{Active}\, (\textit{course})$ with
all active courses
and a table $\textit{Registered}\,(\textit{student},\textit{course})$
that records which students have been registered for which course. 
Suppose we have a process model that does not allow processes to write 
into $\textit{Active}$  and 
which states that before a student
is registered for a course,
there must be a check that the course is active.
Consider the query $Q_{\textit{agro}}$ that asks for all students registered
for the MSc in Agronomics (\textit{mscAgro}). 
If \textit{mscAgro} does not occur in $\textit{Active}$, then no student can be registered
and the query is stable.
Consider next the query $Q_{\textit{courses}}$ that asks for all courses
for which some student is registered.
If for each active course there is at least one student registered,
then again the query is stable,
otherwise, it is not stable because some student could register for 
a so far empty active course.

In general, query results can be affected by the activities of processes in several ways.
Processes may store data from outside in the database,
e.g., the application details submitted by students are stored in the database.
Processes may not proceed because data does not satisfy a required condition,
e.g., an applicant cannot register because his degree is not among the recognized degrees.
Processes may copy data from one part of a database to another one,
e.g., students who passed all exams are automatically registered for the next year.
Processes may interact with each other in that one process writes data that is
read by another one,
e.g., the grades of entry exams stored by the student office
are used by academic admission committees.
Finally, some activities depend on deadlines
so that data cannot change before or  after a deadline.


\paragraph*{Approach}
Assessing query stability by leveraging on processes
gives rise to several research questions.
(1) What is a good model to represent processes,
data and the interplay among the two? 
%
%
(2) How can one reason on query stability in such a model and how feasible is that?
(3) What characteristics of the model may complicate reasoning?
%


(1) \textbf{Monotonic Data-Aware Business Process Model.}
Business processes are often specified
in standardized languages, such as BPMN \cite{bpmn}, and 
organizations rely on engines that can run those processes 
(e.g., Bonita \cite{bonita}, Bizagi \cite{bizartifact}). 
However, in these systems how the data is manipulated by the 
process is implicit in the code.
Current theory
approaches either focus on \emph{process} modeling,
representing the data in a limited way
(like in Petri Nets~\cite{jensen-colouredPN-2009}),
or adopt a \emph{data} perspective, 
leaving the representation of the process flow 
implicit~\cite{artifact-formal-analysis,Calvanese-DCDS-PODS-2013,Deutsch-WebData-PODS-2004}.
We introduce a formalism called Monotonic Data-aware Business Processes (\DABP{s}).
In \DABP{s} the process is represented as a graph.
The interactions with an underlying database are expressed by annotating
the graph with information on which data is read from the database and which
is written into it.
%
In {\DABP{s} it is possible that several process instances execute the process.
New information {(fresh data)} can be brought into the process by starting 
a fresh process instance}
(Section~\ref{sec:dabp}).
%
\DABP{s} are 
monotonic in that data can only be inserted, but not deleted or updated.

(2) \textbf{Datalog Encodings.} 
Existing approaches 
aim at the verification of general (e.g.\ temporal) properties,
for which reasoning is typically
intractable~\cite{Calvanese-DCDS-PODS-2013,ltl:verification:hull:vianu-icdt2009-decidable,Deutsch-WebData-PODS-2004}.
In contrast, we study a specific property, namely
stability of conjunctive queries (Section~\ref{sec:query-stability-problem}),
over processes that only insert data.
This allows us to map the problem to the one of query answering in Datalog.
%
The encoding generates all maximal representative extensions of the database that can be produced 
in the process executions
and checks if any new query answer is produced.
We prove that our approach is optimal w.r.t.~worst case complexity
in the size of the data, query, process model and in the size of the entire input. 

%
%
%
%
%

(3) \textbf{\DABP Variants.}
When modeling processes and data, checking properties often
becomes highly complex or undecidable.
While other approaches in database theory aim at exploring
the frontiers of decidability by restricting the possibility to introduce fresh data,
we adopt a more bottom-up approach and focus on a simpler problem
that can be approached by established database techniques.
To understand the sources of complexity of our reasoning problem, 
we identify \emph{five restrictions} of \DABP{s}: 
\ei negation is (is not) allowed in process conditions;
\eii the process can (cannot) start with pending instances;
\eiii a process can (cannot) have cycles;
\eiv a process can (cannot) read from relations that it can write;
\ev new instances can (cannot) start at any moment.
Combinations of these restrictions define different variants
of \DABP{s}, for which we investigate the stability problem
(Sections~\ref{sec:query-stability-problem}--\ref{sec:rowo}).

Related work and conclusions end the paper 
(Sections~\ref{ses:related:work},~\ref{sec:discussion:conlcusion}).
%
A technical report, with complete encodings and proofs 
can be found in  \cite{DABPS-KRDB:Report}.

A preliminary version of this paper was presented at the AMW workshop~\cite{AMW}.
 
\section{Monotonic Data-Aware Business Processes}
           \label{sec:dabp}
\emph{Monotonic Data-aware Business Processes} (\DABP{s}) are the formalism by
which we represent business processes and the way they manipulate data.
We rely on this formalism to perform reasoning on query stability.

\paragraph*{Notation}
We adopt standard notation from databases.
In particular, we assume an infinite set of relation symbols,
an infinite set of constants $\dom$ as the \emph{domain of values},
and the positive rationals $\RationalsPos$  as the \emph{domain of timestamps}.
A schema is a finite set of relation symbols.
A \emph{database instance} is a finite set of ground atoms, called \emph{facts},
over  a schema and the \emph{domain} $\domQPos = \dom \cup \RationalsPos$.
We use upper-case letters for variables, 
lower-case for constants,
and overline for tuples, e.g., $\tc$.

An \DABP is a pair $\B=\tup{\P,\C}$, 
consisting of a \emph{process model} $\P$ and a \emph{configuration} $\C$.
The process model defines how and under which conditions actions change data 
stored in the configuration.
The configuration is dynamic,
consisting of \ei a database,
and \eii the process instances.

\paragraph*{Process Model}

The process model is a pair $\P=\tup{N,L}$,
comprising a directed multigraph~$N$, the \emph{process net},
and a \emph{labeling function} $L$, defined on the edges of $N$.

The net $N = \tup{P,T}$ consists of a set of vertices $P$, the \emph{places},
and a multiset of edges $T$, the \emph{transitions}.
A process instance traverses the net, starting from the distinguished place $\start$.
The transitions emanating from a place represent
alternative developments of an instance.

A process instance has input data associated with it,
which are represented by a fact $\In(\tc,\tau)$,
where $\In$ is distinguished relation symbol,
$\tc$ is a tuple of constants from $\domQPos$, 
and $\tau \in \RationalsPos$ is a time stamp
that records when the process instance was started.
We denote with $\SigmaBIn$ and $\SigmaB$ the schemas of $\B$
with and without $\In$, respectively.

The {labeling function\/} $L$ assigns to every transition $t\in T$ 
a pair $L(t)=(E_t,W_t)$.
Here, $E_t$, the \emph{execution condition,} 
is a Boolean query over $\SigmaBIn$ and 
$W_t$, the \emph{writing rule,} is a rule
$R(\tu) \la B_t(\tu)$
whose head is a relation of $\SigmaB$
and whose body is a $\SigmaBIn$-query that has the same arity as the head relation.
Evaluating $W_t$ over a $\SigmaBIn$-instance $\D$ results in 
the set of facts 
$W_t(\D) = \set{R(\tpl c) \mid \tpl c \in B_t(\D)}$. 
Intuitively, 
$E_t$ specifies in which state of the database 
which process instance can perform the transition $t$, and 
$W_t$ specifies which new information is (or can be)
written into the database when performing $t$.
In this paper we assume that $E_t$ and $B_t$ are conjunctive queries, 
possibly with negated atoms and inequality atoms with \quotes{$<$}
and \quotes{$\leq$} involving timestamps.
We assume inequalities to consist of one constant and one variable,
like $X < \mbox{\first\ Sep}$.
We introduce these restricted inequalities so that we can model deadlines,
without introducing an additional source of complexity for reasoning.

\paragraph*{Configuration}
This component models the dynamics of an \DABP.
Formally, a configuration is a triple $\tup{\I, \D,\tau}$, consisting of
a part $\I$ that captures the process instances, 
a database instance $\D$ over $\SigmaB$, and
a timestamp $\tau$, the current time.
The instance part, again, is a triple $\I = \tup{ O, M_{\In}, M_P }$,
where $O =\set{o_1, ... , o_k}$ is a set of objects, called \emph{process instances},
and $M_{\In}$, $M_P$ are mappings,
associating each $o\in O$ with
a fact $M_\In(o) = \In(\tc,\tau)$, its input record, and
a place $M_P(o) \in P$, its current, respectively.

The input record is created when the instance starts and cannot be changed later on. 
While the data of the input record may be different from the constants in the database,
they can be copied into the database by writing rules.
A process instance can see the entire database, but only its own input record.

For convenience, we also use the notation
$\B = \tup{ \P, \I, \D, \tau}$,
$\B = \tup{ \P, \I, \D }$ (when $\tau$ is not relevant), or
$\B = \tup{ \P, \D }$ (for a process that is initially without running instances).



\paragraph*{Execution of an \DABP}
Let $\B=\tup{\P,\C}$ be an \DABP,
with current configuration $\C=\tup{\I,\D,\tau}$.
There are two kinds of \emph{atomic execution} steps of an \DABP: 
\ei the \emph{traversal} of a transition 
by an instance
and 
\eii the \emph{start} of a new instance.

\smallskip

\noindent
\textbf{\ei Traversal} of a transition.
Consider an instance $o\in O$ 
with record $M_\In(o)=\In(\bar c,\tau')$, 
currently at place $M_P(o)=q$.
Let $t$ be a transition from $q$ to $p$, with execution condition $E_t$.
Then 
$t$ is \emph{enabled for $o$,} i.e., $o$ can \emph{traverse} $t$, 
if $E_t$ evaluates to $\true$ over the database 
$\D \cup \set{\In(\tc, \tau')}$.
Let $W_t \colon R(\tpl u) \la B_t(\tpl u) $ 
be the writing rule of $t$.
Then the effect of $o$ traversing $t$ is the transition from 
$\C=\tup{\I,\D, \tau}$ to a new configuration $\C'=\tup{\I',\D',\tau}$,
such that
\ei the set of instances $O$ and the current time $\tau$ are the same;
\eii the new database instance is {$\D'= \D \cup W_t(\D \cup \set{\In(\bar c,\tau')})$}, and
\eiii $\I = \langle O, M_\In, M_P \rangle$ is updated to
			$\I' = \langle O, M_\In, M'_P \rangle$ reflecting the
			change of place for the instance $o$, that is 
      $M'_P(o)=p$ and $M'_P(o') = M_P(o')$
      for all other instances $o'$.
\smallskip

\noindent
\textbf{\eii Start} of a new instance. 
Let $o'$ be a fresh instance and let $\In(\bar c',\tau')$ be an $\In$-fact
with $\tau'\geq\tau$,
the current time of $\C$.
%
The result of starting $o'$ with info $\tpl c'$ at time~$\tau'$ is
the configuration 
$\C'=\tup{\I',\D, \tau'}$ where $\I'=\langle O', M'_\In, M'_P \rangle$
such that
\ei the database instance is the same as in $\C$,
\eii the set of instances $O'=O\cup\set{o'}$ is augmented by $o'$, and
\eiii the mappings $M'_\In$ and $M'_P$ are 
extensions of $M_\In$ and $M_P$, resp.,
obtained by defining 
      $M'_\In(o')=\In(\bar c',\tau')$
      and $M'_P(o')=\start$.

\smallskip

An \emph{execution} $\Upsilon$ of $\B=\tup{\P,\C}$
is a finite sequence of configurations $\C_1,\ldots,\C_n$
\ei starting with $\C$ ($= \C_1$), where
\eii each 
$\C_{i+1}$ is obtained from $\C_{i}$ by an atomic execution step.
We denote $\Upsilon$ also with
$\C_1 \rightsquigarrow \dots \rightsquigarrow \C_n$.     
We say that the execution 
$\Upsilon$
\emph{produces} the facts
$A_1,\ldots,A_n$
if the database of the last configuration $\C_n$ in 
$\Upsilon$ contains $A_1,\ldots,A_n$.
Since at each step a new instance can start, or 
an instance can write new data,
\ei there are infinitely many possible executions, and 
\eii the database may grow in an unbounded way over time.


\begin{table*}[!t]
\rra{1.3}
	
	\begin{center}
	\begin{tabular}{@{}llcccccc@{}}
		\toprule


		\begin{minipage}[c]{6.0em}
			{\footnotesize \textbf{Defining}  } \\[-0.15em]
			{\footnotesize \textbf{Restrictions} }
		\end{minipage}		&
		\begin{minipage}[c]{6.0em}
			{\footnotesize \textbf{Optional}  } \\[-0.15em]
			{\footnotesize \textbf{Restrictions} }
		\end{minipage}	
		&  \footnotesize{\textbf{Data}}
		&  \footnotesize{\textbf{Instance}} 
		&  \footnotesize{\textbf{Process}}  	
		&  \footnotesize{\textbf{Query}}
		&  \footnotesize{\textbf{Combined}}	
		&  \footnotesize{\textbf{\hspace*{-1ex}Sect.}}			
		\\
		\midrule

		---
		& \footnotesize{$\textit{fresh}^\dagger$,\textit{acyclic}}
		& {\scriptsize $\undec$}	
		& {\scriptsize $\undec$}
		& {\scriptsize $\undec$}
		& {\scriptsize $\undec$}
		& {\scriptsize $\undec$}   	
		& {\scriptsize \ref{sec:undecidable}} 
		\\
		
				

		\footnotesize{\textit{closed}}
		& ---
		& {\scriptsize  $\coNP$   }	
		& {\scriptsize  $\coNP$   }	
		& {\scriptsize $\coNEXPTIME$}
		& {\scriptsize $\PiPTWO$}
		& {\scriptsize $\coNEXPTIME$}   		
		& {\scriptsize \ref{subsubsec:normal:cyclic:arbitrary}}
		\\

		\footnotesize{\emph{positive}}
		& \footnotesize{\emph{closed}}
		& {\scriptsize  $\PTIME$   }	
		& {\scriptsize  $\coNP$   }	
		& {\scriptsize $\EXPTIME$}
		& {\scriptsize $\PiPTWO$}
		& {\scriptsize $\EXPTIME$}   	
		& {\scriptsize \ref{sec-positive:closed}, \ref{sec:DABP-Datalog-power-short}}   		
		\\   

		\footnotesize{\emph{positive}}
		& \footnotesize{$\textit{fresh}^\dagger$, \textit{acyclic}}
			& {\scriptsize  $\PTIME$ }	
		& {\scriptsize  $\coNP$  }	
		& {\scriptsize $\EXPTIME$}
		& {\scriptsize $\PiPTWO$}
		& {\scriptsize $\EXPTIME$}   	
		& {\scriptsize \ref{sec:DABP-Datalog-power-short}}   		
		\\		

			\footnotesize{\textit{closed, acyclic}}
		& \footnotesize{\textit{positive}}
		& {\scriptsize  in $\ACz$      }	
		& {\scriptsize  $\coNP$  }	
		& {\scriptsize $\PSPACE$}
		& {\scriptsize $\PiPTWO$}
		& {\scriptsize $\PSPACE$}   			
		& {\scriptsize \ref{sec:acyclic:arbitrary}} 
		\\
		\footnotesize{\emph{rowo}}
		&	\footnotesize{${\ast}^\dagger$}
		& {\scriptsize  in $\ACz$      }	
		& {\scriptsize  in $\ACz$   }	
		& {\scriptsize $\coNP$}
		& {\scriptsize $\PiPTWO$}
		& {\scriptsize $\PiPTWO$}   	
		& {\scriptsize \ref{sec:rowo}} 
		\\
		\bottomrule

	\end{tabular}
	\end{center}
	\vspace*{-2ex}
	\caption{Computational complexity of query stability in \DABP{s}.
                 The results in a row hold for the class of \DABP{s} satisfying 
                 the defining restrictions 
                 and for the subclasses satisfying one
                 or more of the optional restrictions. 
 		 The results for all decidable variants indicate matching lower and upper bounds
                 (except for $\ACz$).
		 The $\ast$ indicates that the results for rowo hold for all non-trivial combinations of restrictions. 
		 All results for data, process, query and combined complexity of the decidable variants
		 hold already for singleton \DABP{s}.
		${^\dagger}$Note that, in all fresh variants instance complexity can be trivially decided in constant time (omitted in the table).
		}
	\label{table:complexityDABP-overall:new2}
	 \vspace*{-3ex}
\end{table*}%


\section{The Query Stability Problem}
                \label{sec:query-stability-problem}

In this section, we define the problem of query stability in \DABP{s}   
with its variants.

\begin{definition}[Query Stability]
  \label{def:query-stability}
Given $\B=\tup{\P,\C}$ with database instance $\D$,
a query $Q$, and a timestamp $\tau$,
we say that \emph{$Q$ is stable in $\B$ until $\tau$},
if for every execution 
$\C \rightsquigarrow \cdots \rightsquigarrow \C'$ in $\B$,
where $\C'$ has database $\D'$ and timestamp  $\tau'$ 	
such that $\tau'<\tau$,
it holds that 
\[ 
		Q(\D)=Q(\D').
\] 
If the query $Q$  is stable until time point $\supremum$,
we say it is \emph{globally stable}, 
or simply, \emph{stable}.
\end{definition}

The interesting question from an application view is:
\emph{Given an \DABP $\B$, a query $Q$, and a timestamp $\tau$, 
      is $Q$ stable in $\B$ until $\tau$?} 
Stability until a time-point $\tau$ can be reduced to global stability.
One can modify a given \DABP by adding a new \start\ place and 
connecting it to the old \start\ place via a transition 
that is enabled only for instances with timestamp smaller than $\tau$.
Then a query $Q$ is globally stable in the resulting \DABP iff 
in the original \DABP it is stable until $\tau$.

To investigate sources of complexity and provide suitable encodings into
Datalog, we identify five restrictions on \DABP{s}.


\begin{definition}[Restriction on \DABP{s} and \DABP Executions]
Let $\B$ be an \DABP.


\begin{description}
\item[Positive:]
$\B$ is \emph{positive} if execution conditions and writing rules
contain only positive atoms;
\item[Fresh:]
$\B$ is \emph{fresh} if its configuration does not 
contain any running instances;
\item[Acyclic:]
$\B$ is \emph{acyclic} if the process net is cycle-free;
\item[Rowo:]
$\B$ is \emph{rowo} (= read-only-write-only)
if the schema $\Sigma$ of $\B$ can be split into two disjoint schemas:
the reading schema $\Sigma_r$ and the writing schema $\Sigma_w$,
such that execution conditions and 
queries in the writing rules range over $\Sigma_r$ 
while the heads range over $\Sigma_w$;
\item[Closed:] 
an execution of $\B$ is \emph{closed} 
if it contains only transition traversals and no new instances are started.
\end{description}
\end{definition}

\vspace{-1.5ex}

We will develop methods for stability checking in \DABP{s} 
for all combinations of those five restrictions.
For convenience, we will say that an \DABP $\B$ is \emph{closed\/} if we consider only 
closed executions of $\B$.
A \emph{singleton \DABP} is a closed \DABP
with a single instance in the initial configuration.

\paragraph*{Complexity Measures}
The input for our decision problem are 
an \DABP $\B=\tup{\P,\I,\D}$, consisting of
a process model $\P$, 
an instance part $\I$,
a database $\D$ and
a timestamp $\tau$, and
a query $Q$.
The question is:
	\textit{Is $Q$ globally stable in  $\tup{\P,\I, \D, \tau}$?} 
We refer to process, instance, data, and query complexity if all parameters are fixed,
except the process model, the instance part, the database, or the query, respectively.

\paragraph*{Roadmap}
As a summary of our results,
Table~\ref{table:complexityDABP-overall:new2} 
presents the complexity of the possible variants of query stability.
Each section of the sequel will cover one row.

\paragraph*{Datalog Notation}
We assume familiarity with Datalog
concepts such as 
\emph{least fix point} and
\emph{stable model} semantics,
and \emph{query answering}
over Datalog programs under both
semantics.
We consider Datalog programs that are
\emph{recursive}, 
\emph{non-recursive}, 
\emph{positive},
\emph{semipositive},
\emph{with negation}, or
\emph{with stratified negation}~\cite{Gottlob-Eiter-LP-Journal-2001}.
%
%
We write $\Pi \cup \D$ to denote a program 
where 
$\Pi$ is a set of rules and
$\D$ is a set of facts.

\paragraph*{Summary of notation} 
For convenience, we summarize the notation of our model used in the
following sections in Table~\ref{tb:notation-table}.



\begin{table*}[!t]
	
	\begin{center}
	\begin{tabular}{@{}ll@{}}

\toprule


\textbf{Notation} & \textbf{Meaning}\\
\midrule

$\B=$
$\langle \P, \C \rangle$,
$\langle \P, \I, \D  \rangle$,
$\langle \P, \I, \D, \tau  \rangle$
	& \DABP \\

$\P=\tup{N,L}$
&
Process part $\P$ with process net $N$ and labeling function $L$
\\

$\C=\tup{\I,\D,\tau}$
&
Configuration  $\C$ with instance part $\I$, 
\\
\quad & database instance $\D$ and timestamp $\tau$
\\

$N = \tup{P,T}$
& 
Process net $N$ with places $P$ and transitions $T$
\\

$T = \set{t_1,\dots,t_m}$
& 
Multi set of process transitions
\\

$p,q$
& 
Process places
\\

$L(t)=(E_t,W_t)$
&
Labeling $L(t)$ of transition $t$ 
\\

$\In(\ts,\tau), \In(\ts)$ 
& 
Input records with and without timestamps
\\

$E_t \colon \In(\ts), S_1,\dots,S_n,G_t$
&
Execution condition with possibly negated atoms  $S_i$ and
\\
\quad & 
conjunction of inequalities $G_t$ 
\\

$W_t \colon R(\tu) \la B_t(\tu)$ 
& 
Writing rule with query $B_t(\tu) \la \In(\ts), S_1, \dots, S_l, M_t$ with
\\
\quad & 
possibly negated atoms $S_i$  and
conjunction of inequalities $M_t$
\\

$\I = \tup{O, M_\In, M_P}$
& 
Instance part with process instances $O = \set{o_1,\dots,o_k}$
\\


$\Upsilon\colon \C_1 \rightsquigarrow \dots \rightsquigarrow \C_n$ 
& 
\DABP execution with configurations $\C_1, \dots, \C_n$
\\

$r, a, k, m, c$ 
& \DABP $\B$ has $r$ relations with $a$ as the maximal arity, 
\\
\quad
&
$k$ running instances, $m$ transitions,
	$c$ number of constants\\

\bottomrule

\end{tabular}
\end{center}
\caption{Notation table of symbols that represent \DABP{s}}
\label{tb:notation-table}

\end{table*}

\section{Undecidable \DABP{s}}
\label{sec:undecidable}

With negation in execution conditions and writing rules,
we can create \DABP{s} that simulate Turing machines (TMs).
Consequently, in the general variant query stability is undecidable.

Due to lack of space we only provide an intuition.
To show undecidability in data complexity, 
we define a database schema that allows us to store a TM
and
we construct a process model that simulates the executions of the stored TM.
\DABP{s} cannot update facts in the database.
However, we can augment relations with an additional version argument
and simulate updates by adding new versions of facts.
Exploiting negation in conditions and rules 
we can then refer to the last version of a fact.
To simulate the TM execution, the process model
uses fresh constants to model
\emph{(i)}~an unbounded number of updates of the TM configurations (= number of execution steps in the TM), and 
\emph{(ii)}~a potentially infinite tape.
The TM halts iff the process produces the predicate $\dummy$. 
Undecidability in process complexity follows
from undecidability in data complexity, 
since a process can first write the encoding of the TM into 
an initially empty database.
Similarly, we obtain undecidability in instance complexity using
instances that write the encoding of the TM at the beginning.
%
To obtain undecidability in query complexity we 
extend the encoding for data complexity such that 
the database encodes a universal TM and
an input of the TM is encoded in the query.
%


\begin{theorem}[Undecidability]
  Query stability in \DABP{s} is undecidable 
  in
  data,
  process 
  and query complexity.
  It is also undecidable in instance complexity 
  except for fresh variants
  for which it is constant.
  Undecidability already holds for  acyclic
  \DABP{s}.%
\end{theorem}

In our reduction it is the unbounded number of fresh instances 
that are causing writing rules to be executed an unbounded number of times,
so that neither cycles nor existing instances are contributing to undecidability.
In the sequel we study \DABP{s} that are 
positive, closed, or rowo, and show that in all three variants stability is decidable.



\section{Positive Closed \DABP{s}}
  \label{sec-positive:closed}

%
In cyclic positive \DABP{s}, executions can be arbitrarily long.
Still, in the absence of fresh instances, 
it is enough to consider executions of bounded length to check stability.
%
Consider a positive \DABP $\B = \tup{\P,\C}$, 
possibly with cycles and disallowing fresh instances
to start,
with
$c$ different constants,
$r$ relations, 
$k$ running instances,
$m$ transitions
and $a$ as the maximal arity of a relation in $\P$.
We observe:
\begin{inparaenum}[\itshape (i)]
\item  
For each relation $R$ in $\P$
there are up to $c^{{}\arity(R)}$ new $R$-facts that $\B$ can produce. 
Thus, $\B$ can produce up to $rc^a$ new facts in total.
\item
It is sufficient to consider executions that produce at least one new fact each $mk$ steps.
An execution that produces no new facts in $mk$ steps has at least one instance 
that in those $mk$ steps visits the same place twice 
without producing a new fact; 
those steps can be canceled 
without affecting the facts that are produced.
\item 
Hence,
it is sufficient to consider
  executions of maximal length 
$m k r c^a.$
%
\end{inparaenum}

Among these finitely many executions, it is enough to consider 
those 
that produce a maximal set of new facts. 
Since a process instance may have the choice among several transitions,
there may be several such maximal sets.
We identify a class of executions
in positive closed \DABP{s}, 
called \emph{greedy executions},
that produce all maximal sets.

\paragraph*{Greedy Executions}
Intuitively, in a greedy execution instances 
traverse all cycles in the net in all possible ways
and produce all that can be produced before leaving the cycle.
To formalize this idea we identify two kinds of execution steps:
\emph{safe steps} and 
\emph{critical steps}.
A safe step is an execution step of an instance after which,
given the current state of the database,
the instance can return to its original place.
A critical step is an execution step that is not safe.
Based on this, we define \emph{greedy sequences} and \emph{greedy executions}.
A greedy sequence is a sequence of safe steps  
  that produces the largest number of new facts possible.
A greedy execution is an execution where 
greedy sequences and critical steps alternate.

Let $\Upsilon$ be a greedy execution with $i$
alternations of greedy sequences and critical steps. 
In the following, 
we characterize which are the transitions that instances 
traverse in the  $i+1$-th greedy sequence and then in the $i+1$-th
critical step.
For a process instance $o$ and the database $D_{\Upsilon}$ produced after 
$\Upsilon$ we define the \emph{enabled graph} $N_{\Upsilon,o}$ as the multigraph 
whose vertices are the places of $N$ (i.e., the process net of $\B$) 
and edges those transitions of $N$ that are enabled for $o$ given database $D_{\Upsilon}$.
Let $\SCC(N_{\Upsilon,o})$ denote the set of 
strongly connected components (SCCs) of 
$N_{\Upsilon,o}$.
Note that two different instances may have different enabled graphs and thus different SCCs.
For a place $p$,
let $N^p_{\Upsilon,o}$ be the SCC in $\SCC(N_{\Upsilon,o})$
that contains $p$.
Suppose that $o$ is at place $p$ after $\Upsilon$. 
Then in the next greedy sequence, 
each instance $o$ traverses the component $N^p_{\Upsilon,o}$ in all possible ways until no new facts can be produced, 
meaning that all instances traverse in an arbitrary order.
%
Conversely, the next critical step is 
an execution step where an instance $o$ traverses a transition 
that is not part of $N^p_{\Upsilon,o}$, and thus it 
leaves the current SCC.
We observe that when performing safe transitions new facts may be written and new transitions may become executable. This can make SCCs of $N_{\Upsilon,o}$ to grow and merge, enabling new safe steps. 
With slight abuse of notation we denote such maximally expanded SCCs with $N_{\Upsilon,o}$, and 
with $N^p_{\Upsilon,o}$ the maximal component that contains $p$.

\paragraph*{Properties of Greedy Executions}
We identify three main properties of greedy executions.
\begin{itemize}
\item  
A greedy execution is characterized by its critical steps,
because an instance may have to choose one among several possible critical steps.
In contrast, 
how safe steps compose a greedy sequence 
is not important for stability
because all greedy sequences 
produce the same (maximal) set of facts.
\item
A greedy execution in an \DABP
with $m$ transitions and $k$ instances
can have at most $mk$ critical steps.
The reason is that an execution step can be critical only 
the first time it is executed, and 
any time after that it will be a safe step.
\item
Each execution can be transformed into  a greedy execution
such that 
if  a query is instable in the original version 
then it is instable also in the greedy version.
In fact,
an arbitrary execution 
has at most $mk$ critical steps.
One can construct a greedy version starting from those critical steps,
such that the other steps are part of the greedy sequences.
\end{itemize}

\begin{lemma}\label{lemma:reduction:greedy:execution:extended}
  For each closed execution $\Upsilon$ in a positive \DABP $\B$ 
  that produces the set of ground atoms $W$,
  there exists a greedy execution $\Upsilon'$ in $\B$
  that also produces  $W$.
\end{lemma}

Therefore,
to check stability it is enough to check stability over greedy executions.
In the following we define Datalog rules that compute facts produced by greedy executions.

\paragraph*{Encoding into Datalog}
Let $\BPID$ be a positive \DABP 
with $m$ transitions and $k$ instances.
Since critical steps uniquely characterize a greedy execution, 
we use a tuple of size up to $mk$ to encode them.
For example, if in a greedy execution $\Upsilon$ at the first critical step
instance $o_{l_1}$ traverses transition $t_{h_1}$, in the second $o_{l_2}$ traverses $t_{h_2}$, and so on up to step $i$, we encode this with the 
tuple 
\[\tomega = \tup{o_{l_1},t_{h_1},\dots,o_{l_i},t_{h_i}}.\]

Next, we define the relations used in the encoding.
\begin{inparaenum}[\itshape (i)]
\item 
For each relation $R$ in $\P$
we introduce relations $R^i$  
(for $i$ up to $mk$) to store all
$R$-facts produced by an execution 
with $i$ critical steps.
Let $\Upsilon$ be the execution from above 
and let $\tup{o_{l_1},t_{h_1},\dots,o_{l_i},t_{h_i}}$ be the tuple representing it.
Then, a fact of relation $R^i$ has the form 
$
R^i(o_{l_1},t_{h_1},\dots,o_{l_i},t_{h_i};\ts),
$
and it holds iff 
$\Upsilon$ produces the fact $R(\ts)$.
Later on we use $\tomega$ to represent the tuple $\langle o_{l_1},t_{h_1},\dots \rangle$. 
Facts of $R^i$ are then represented as $R^i(\tomega;\ts)$.
For convenience,  we use a semicolon ($;$) instead of a comma ($,$) to separate 
encodings of different types in the arguments. 
\item
To record the positions of instances after each critical step 
we introduce relations $\State^i$ 
such that  $\State^i(\tomega;p_1,\dots,p_k)$ encodes that 
after $\Upsilon$ is executed, 
instance $o_1$ is at $p_1$,
$o_2$ is at $p_2$,
and so on. 
\item
To store the SCCs of the enabled graph we introduce relations $\SCC^i$
such that for a process instance $o$ and a place $p$, 
the transition $t$ belongs to $N^{p}_{\Upsilon,o}$ 
iff $\SCC^i(\tomega;o,p,t)$ is true.
\item
To compute the relations $SCC^i$, we first need to compute which places are reachable by an instance $o$ from place $p$.
For that we introduce auxiliary relations 
$
\Reach^i
$ 
such that in the enabled graph  $N_{\Upsilon,o}$ instance $o$ can reach place $p'$ from $p$ iff $\Reach^i(\tomega;o,p,p')$ is true.
\item 
Additionally, we introduce the auxiliary relation 
$\Inz$ that associates instances with their $\In$-records, 
that is $\Inz(o;\ts)$ is true iff the instance $o$ has input record $\In(\ts)$.
With slight abuse of notation, we use $\tomega$ to denote also the corresponding greedy closed execution $\Upsilon$.%
\end{inparaenum}

In the following we define a Datalog program that computes the predicates introduced above for all possible greedy executions. 
The program uses stratified negation.

\paragraph*{Initialization}
For each relation $R$ in $\P$ we introduce the
\emph{initialization rule} $R^0(\tX) \la R(\tX)$ to store what holds before any critical step is made.
Then we add the fact rule $\State^0(p_1,\dots,p_k)\la \true$ 
if in the initial configuration
$o_1$ is at place $p_1$, $o_2$ at $p_2$, and so on.

\paragraph*{Greedy Sequence: Traversal Rules}
Next, we introduce rules that compute enabled graphs.
%
%
The relation $\Reach^i$ contains the transitive closure of 
the enabled graph $N_{\tomega,o}$ 
for each $o$ and~$\tomega$, encoding a greedy execution of length $i$.
First, a transition $t$ from $q$ to be $p$ gives rise to an edge
in the enabled graph $N_{\tomega,o}$ if instance $o$ can traverse that 
$t$:
\[
  \Reach^i(\tW;O,q,p) \la E^i_t(\tW;O).
\]
Here, 
$E^i_t(\tW;O)$ is a shorthand for the condition obtained from $E_t$ by replacing 
$\In(\ts)$ with $\Inz(O;\ts)$ and by replacing each atom 
$R(\tpl v)$ with $R^i(\tW;\tpl v)$.
The tuple $\tW$ consists of $2i$ many distinct variables
to match every critical execution with $i$ steps. It ensures that only
facts produced by $\tW$ are considered.
The transitive closure is computed with the following rule:
\begin{align*}
  \Reach^i(\tW;O,P_1,P_3) \la 
    \Reach^i(\tW;O,P_1,P_2), \Reach^i(\tW;O,P_2,P_3).  
\end{align*}
Based on $\Reach^i$, $\SCC^i$ is computed by including every transition $t$
from $q$ to $p$ that an instance can reach, traverse, and from where it can return to the current place:
\begin{align*}
    \SCC^i(\tW;O,P,t) \la \ \Reach^i(\tW;O,P,q), E^i_t(\tW;O), \Reach^i(\tW;O,p,P).
\end{align*}

\paragraph*{Critical Steps: Traversal Rules}
We now want to record how an instance makes a critical step.
An instance $o_j$ can traverse transition $t$ from $q$ to $p$
at the critical step $i+1$
if  
\ei $o_j$ is at some place in  $N^{q}_{\tomega,o_j}$ at step $i$,
\eii it satisfies the execution condition $E_t$,
\eiii and by traversing $t$ it leaves  the current SCC.
The following \emph{traversal rule} captures this:
\begin{align}
 &  \State^{i+1}(\tW, o_j, t; P_1,\dots,P_{j-1},p,P_{j+1},\dots,P_k) \la \nonumber \\ 
    & \quad\State^i(\tW;P_1,\dots,P_{j-1},P,P_{j+1},\dots,P_k), 
    \Reach^i(\tW;o_j,P,q), \Reach^i(\tW;o_j,q,P), \label{eq:criticalstepone} \\ 
    & \quad E^i_t(\tW;o_j), \neg \SCC^i(\tW;o_j,P,t). \label{eq:criticalsteptwo}
\end{align}
Here, the condition \ei is encoded in
line \eqref{eq:criticalstepone},
and \eii and \eiii are encoded in line \eqref{eq:criticalsteptwo}.

\paragraph*{Generation Rules}
A fact in $R^{i+1}$ may hold because 
\ei it has been produced by the current greedy sequence or by the last critical step, or 
\eii by some of the previous sequences or steps.
Facts produced by previous sequences or steps 
are propagated with the \emph{copy rule}:
$R^{i+1}(\tW,O,T;\tX) \la \State^{i+1}(\tW,O,T;\blank),R^{i}(\tW;\tX)$,
copying facts $R(\tX)$ holding after $\tW$ to all extensions of $\tW$.

Then we compute the facts produced by the next greedy sequence.
For each instance $o_j$,
being at some place $p_j$ after the last critical step in $\tomega$, 
and for each transition $t$ that is in $N^{p_j}_{\tomega,o_j}$,
with writing rule  
$R(\tu) \la B_t(\tu)$, 
we introduce the following \emph{greedy generation rule}:%
\begin{align*}
    R^{i}(\tW;\tu) \la \ 
          \State^{i}(\tW;\blank,\dots,\blank,P_j,\blank,\dots,\blank),
          \SCC^i(\tW; o_j,P_j,t), B^i_t(\tW;o_j;\tu),
\end{align*}
where condition $B^i_t(\tW;O;\tu)$ is obtained 
similarly as
$E^i_t(\tW;O)$.
In other words, all transitions $t$ that are in $N^{p_j}_{\tomega,o_j}$
are fired simultaneously, and this is done for all instances.

The facts  produced at the next critical step 
by traversing $t$, 
which has the writing rule $R(\tu) \la B_t(\tu)$, 
are generated with the \emph{critical generation rule}:
 $   R^{i+1}(\tW,O,t;\tu) \la  \State^{i+1}(\tW,O,t;\blank),B^i_t(\tW;O;\tu).$

Let $\PiPIpocl$ be the program encoding the positive closed
$\BPID$ as described above.

\begin{lemma} \label{lem:encoding:positive:closed:cyclic:extended}
Let $\tomega$ be a greedy execution in the positive closed 
$\BPID$ of length $i$ and  
$R(\ts)$ be a fact.
Then 
$R(\ts)$ is produced by $\tomega$
\ iff \ 
$\PiPIpocl \cup \D \models R^i(\tomega;\ts)$.
\end{lemma}

%

\paragraph*{Test Program}
Now we want to test the stability of 
$Q(\tX) \la R_1(\tu_1),\dots,R_n(\tu_n)$.
We collect all potential $Q$-answers using the relation $Q'$.
A new query answer may be produced by an execution of any size $i$ 
up to $mk$. 
Thus, for each execution of a size $i$ from $0$ to $mk$ we introduce the $Q'$-rule
\begin{equation}
   \label{eqn-Q:prime:test:rule} 
Q'(\tX) \la R_1^i(\tW;\tu_1),\dots,R_n^i(\tW;\tu_n).
\end{equation}
Then, if there is a new query answer, the \emph{test rule} 
``$
\instable \la Q'(\tX), \neg Q(\tX)
$''
fires the fact $\instable$.
%
Let $\PiPIQtest$ be the test program that contains
$Q$, the $Q'$-rules, 
and the test rule.

\begin{theorem} \label{col:encoding:closed:acyclic:checking}
     $Q$ is instable in the positive closed $\B$
     \ iff\ 
    $\PiPIpocl \cup \D \cup \PiPIQtest \models \instable$.
\end{theorem}


\paragraph*{Data and Process Complexity}%
Since $\PiPIpocl \cup \D \cup \PiPIQtest $ is a Datalog program with stratified negation, 
for which reasoning is as complex as for positive Datalog,
we obtain as upper bounds $\EXPTIME$ for process and combined complexity, and $\PTIME$ for data 
complexity~\cite{Gottlob-Eiter-LP-Journal-2001}.
We show that these are also lower bounds, even for singleton \DABP{s}.
This reduction can also be adapted for acyclic fresh \DABP{s},
which we study in Section~\ref{sec:DABP-Datalog-power-short}. 

\begin{lemma}
 \label{lem:positive:cyclic:complexity}
Stability is $\EXPTIME$-hard in process and 
$\PTIME$-hard in data complexity for
\begin{enumerate}[label=\alph*)]
  \item positive singleton \DABP{s} under closed executions, and
  \item positive acyclic fresh \DABP{s}.
\end{enumerate}
\end{lemma}

\begin{proof}[Proof Sketch]
\emph{a)}
We encode query answering over a Datalog program $\Pi \cup \D$ into stability checking. 
Let $A$ be a fact.
We construct a positive singleton \DABP $\tup{\P_{\Pi,A}^{\textit{po,cl}},\I_0,\D}$,
where there is a transition for each rule and 
the single process cycles to produce the least fixed point (LFP) of the program. 
In addition, the \DABP inserts the fact $\dummy$ if $A$ is in the LFP. 
Then test query $\Qtest \la \dummy$  is stable in $\tup{\P_{\Pi,A}^{\textit{po,cl}},\I_0,\D}$
iff $\Pi \cup \D \not\models A$.

\emph{b)} Analogous, letting fresh instances play the role of the cycling singleton instance.
\end{proof}

\paragraph*{Instance Complexity}
Instance complexity turns out to be higher than data complexity. 
Already for acyclic positive closed \DABP{s} it is $\coNP$-hard because
\begin{inparaenum}[\itshape (i)]
\item process instances may non-deterministically choose a transition, 
      which creates exponentially many combinations, even in the acyclic variant; and
\item instances may interact by reading data written by other instances. 
\end{inparaenum}

\begin{lemma}\label{lem:instance:complexity:closed:acyclic}
There exist a positive acyclic process model $\P_0$, a database $\D_0$,
and a test query $\Qtest$ with the following property:
for every graph $G$ one can construct an instance part $\I_G$ such that
$G$ is not 3-colorable
iff
$\Qtest$  is stable in $ \tup{\P_0,\I_G,\D_0}$ under closed executions.
\end{lemma}

Clearly, Lemma~\ref{lem:instance:complexity:closed:acyclic} implies that 
checking stability for closed \DABP{s} is $\coNP$-hard in instance complexity.
According to Theorem~\ref{theorem:sms} (Section~\ref{subsubsec:normal:cyclic:arbitrary}),
instance complexity is $\coNP$ for all closed \DABP{s},
which implies $\coNP$-completeness even for the acyclic variant.

\paragraph*{Query Complexity}
To analyze query complexity we first show how difficult it is to check 
whether a query returns the same answer over a database and an extension
of that database.

\begin{lemma}[Answer Difference]
\label{lem-cq:piptwo}
For every two fixed databases $\D \subseteq \D'$,
checking whether a given conjunctive query $Q$ satisfies $Q(\D)=Q(\D')$ is in $\PiPTWO$ in the query size.
Conversely, there exist databases $\D_0 \subseteq \D'_0$
such that checking for a conjunctive query $Q$ whether $Q(\D_0)=Q(\D'_0)$ 
is $\PiPTWO$-hard in the query size.
\end{lemma}

\begin{proof}[Proof Idea]
The first claim holds since one can check 
$Q(\D) \subsetneqq Q(\D')$ 
in $\NP$ using an $\NP$~oracle.
We show the second by reducing the \emph{3-coloring extension} 
problem for graphs~\cite{Fagin-3ColorExtension}.
\end{proof}

Building upon Lemma~\ref{lem-cq:piptwo}, 
we can define an \DABP that starting from $\D_0$ produces $\D'_0$.
In fact, for such an \DABP it is enough to consider the simplest  variants of rowo.

%

\begin{proposition}
 \label{prop-hardness:query:complexity}
 Checking stability is $\PiPTWO$-hard for 
\begin{enumerate}[label=\alph*)]
  \item  positive fresh acyclic rowo \DABP{s}, and
  \item  positive closed acyclic rowo singleton \DABP{s}.
\end{enumerate} 
\end{proposition}

Given $\B = \tup{\P,\I,\D}$, there are finitely many maximal extensions $\D'$ of $\D$
that can be produced by $\B$.
We can check stability of a query $Q$ by finitely many checks whether $Q(\D) = Q(\D')$.
Since each such check is in $\PiPTWO$, according to Lemma~\ref{lem-cq:piptwo},
the entire check is in $\PiPTWO$.
Thus, stability is $\PiPTWO$-complete in query complexity.


\section{Closed \DABP{s}}
\label{subsubsec:normal:cyclic:arbitrary}
%

In the presence of negation, inserting new facts may disable transitions.
During an execution, a transition may switch many times between being enabled and disabled,
and greedy executions could have exponentially many critical steps.
An encoding along the ideas of the preceding section would lead to a program of exponential size.
This would give us an upper bound of double exponential time for combined complexity.
%
%
Instead, we establish a correspondence between stability and brave query answering for Datalog with (unstratified) negation under
\emph{stable model semantics} (SMS) 
\cite{Gottlob-Eiter-LP-Journal-2001}.
Due to lack of space we only state the results. 

\begin{theorem}  \label{theorem:sms}
For every closed \DABP $\B=\tup{\P,\I,\D}$ 
and every query $Q$ 
one can construct 
a Datalog program with negation $\Picl$, based on $\P$,
a database $\D_\I$, based on $\D$ and $\I$, 
and a test program $\PiQtest$, based on $Q$,
such that the following holds: \\[.7ex]
\centerline{%
$Q$ is instable in $\BPID$ 
\quad iff\quad
$\Picl \cup \D_\I \cup \PiQtest \mdl \instable$.}
\end{theorem}
\begin{proof}[Proof Idea]
For the same reason as in the positive variant,
it is sufficient to consider executions of maximal length 
$mkrc^a$.
Program $\Picl$ contains two parts: 
\ei a program 
that generates a linear order of size $mkrc^a$ 
(with parameters $m, k, r, c, a$ defined as in Section~\ref{sec-positive:closed}),
starting from an exponentially smaller order,
that is used to enumerate execution steps, and
\eii a program that \quotes{guesses} an execution of size up to $mkrc^a$
by selecting for each execution step  one instance and one transition, 
and that produces the facts that would be produced by the guessed execution.
Then each execution corresponds to one stable model.
The test  program $ \PiQtest$ checks if any of the guessed executions yields a new query answer.
\end{proof}
In  Theorem~\ref{theorem:sms}, the process is encoded in the program rules
while data and instances are encoded as facts.
Since brave reasoning under SMS is $\NEXPTIME$ in program size and $\NP$ in data size~\cite{Gottlob-Eiter-LP-Journal-2001},
we have that process and combined complexity are in $\coNEXPTIME$,
and data and instance complexity are in $\coNP$.
%
From this and Lemma~\ref{lem:instance:complexity:closed:acyclic}
it follows that instance complexity is $\coNP$-complete.
To show that stability is 
$\coNEXPTIME$-complete in process and $\coNP$-complete in data
complexity we encode brave reasoning into stability.
Query complexity is $\PiPTWO$-complete for the same reasons as in the positive variant.

\begin{theorem}
\label{proposition:cy-cl-process-hd:extended}
For every Datalog program $\Pi \cup \D$, possibly with negation, and fact~$A$, 
one can construct a singleton \DABP $\tup{\P_{\Pi,A},\I_0,\D}$
such that for the query $\Qtest\la\dummy$ we have:
$\Pi \cup D \mdl A$ \ iff \ $\Qtest$ is stable in $\tup{\P_{\Pi,A},\I_0,\D}$
under closed executions.
\end{theorem}




\section{Acyclic Closed \DABP{s}} \label{sec:acyclic:arbitrary}

If a process net is cycle-free, all closed executions have finite length.
More specifically, in an acyclic \DABP with $m$ transitions and $k$ running instances,
the maximal length of an execution is $mk$.
%
%
Based on this observation, we modify the encoding for the positive closed variant in Section~\ref{sec-positive:closed}
so that it can cope with negation and exploit the absence of cycles.

For an acyclic \DABP, there cannot exist any greedy steps, which would stay in a strongly connected component of the net.
Therefore, we drop the encodings of greedy traversals and the greedy generation rules.
We keep the rules for critical steps, but drop the atoms of relations $\Reach^i$ and $\SCC^i$.
Now, in contrast to the positive closed variant, we may have negation in 
the conditions $E_t$ and $B_t$.
However, the modified Datalog program is non-recursive, since each relation $R^i$ and $\State^i$ 
is defined in terms of $R^j$'s and $\State^j$'s where $j<i$.

Let $\PiPIQacyccl$ be the program encoding an acyclic $\BPID$ as described above
and let  $\PiPIQtest$ be the test program as in the cyclic variant.

\begin{theorem} \label{theorem:encoding:closed:acyclic:checking}
		$Q$ is instable in the closed acyclic $\B$ 
		\ iff \quad
		$\PiPIQacyccl \cup \D \cup \PiPIQtest \models \instable$.
\end{theorem}

\paragraph*{Complexity}
As upper bounds for combined and data complexity,
the encoding gives us the analogous bounds for non-recursive $\DATALOGNEG$ programs,
that is, $\PSPACE$ in combined and $\ACz$ in data complexity
\cite{Gottlob-Eiter-LP-Journal-2001}. 
Already in the positive variant, 
we inherit $\PSPACE$-hardness of process complexity (and therefore also of combined complexity)
from the program complexity of non-recursive \DATALOG.
We obtain matching lower bounds by a reverse encoding.

\begin{lemma} \label{lem:acyclic:complexity}
For every non-recursive Datalog program $\Pi$ and every fact $A$,
one can construct a singleton acyclic positive \DABP $\tup{\P_{\Pi,A},\C_0}$ such that 
for the query $\Qtest\la\dummy$ we have:
$\Pi \not \models A$ iff  $\Qtest$ is stable in $\tup{\P_{\Pi,A}, \C_0}$ under closed executions.
\end{lemma}

We observe that for closed executions, the cycles increase the complexity, 
and moreover, cause a split between variants with and without negation.
Lemma~\ref{lem:instance:complexity:closed:acyclic} and 
Theorem~\ref{theorem:sms} together imply
that instance complexity is $\coNP$-complete.
Query complexity is $\PiPTWO$-complete for the same reasons as in other closed variants.



\section{Positive Fresh \DABP{s}} 
  \label{sec:DABP-Datalog-power-short}

All decidable variants of \DABP{s} that we investigated until now 
were so because we allowed only closed executions.
In this and the next section we show that decidability 
can also be guaranteed if conditions and rules are positive,
or if relations are divided into read and write relations (rowo).
We look first at the case where initially there are no running instances.

When fresh instances start, their input can bring an arbitrary 
number of new constants into the database.
Thus, processes can produce arbitrarily many new facts. 
First we show how 
infinitely many executions of a positive or rowo \DABP 
can be faithfully abstracted to finitely many over 
a simplified process such that 
a query is stable over the original process iff 
it is stable over the simplified one.
%
%
For such simplified positive \DABP{s}, we show how to encode stability checking into query answering in Datalog.

\paragraph*{Abstraction Principle}
  
Let $\B=\tup{\P,\I,\D,\tauB}$ be a positive or rowo \DABP
and let $Q$ be a query that we want to check for stability.
Based on $\B$ and $Q$ 
we construct an \DABP $\B'=\tup{\P',\I,\D,\tauB}$ 
that has the same impact on the stability of $Q$
but uses at most linearly many fresh values from the domain.
%

%
Let $\adom$ be the active domain of $\B$ and $Q$, 
that is the set of all constants appearing in $\B$ and $Q$.
Let $\tau_1,\dots,\tau_n$ be all timestamps including $\tauB$
that appear in comparisons in $\B$ such that $\tau_i<\tau_{i+1}$. 
We introduce $n+1$ 
many fresh timestamps $\tau'_0,\ldots,\tau'_n \not \in \adom$ such that
$
\tau_0'<\tau_1<\tau_1'< \dots < \tau_n < \tau_n'.
$
If there are no comparisons in $\B$ we introduce one fresh timestamp $\tau'_0$.
Further, let $a$ be a fresh value such that $a \not \in \adom$.
Let $\adomExt = \adom \cup \set{\tau'_0,\ldots,\tau'_n} \cup \set{a}$ 
be the \emph{extended active domain}.

%
Then, 
we introduce the \emph{discretization} function
$\delta_{\B} \col \domQPos \to \domQPos$
that based on $\adomExt$
\quotes{discretizes} 
$\domQPos$ 
as follows:
for each $\tau \in \RationalsPos$
\ei 
$\delta_{\B}(\tau)= \tau$ if $\tau =\tau_i$ for some $i$;
\eii
$\delta_{\B}(\tau)= \tau_i'$ if $\tau_i<\tau< \tau_{i+1}$ for some $i$;
\eiii
$\delta_{\B}(\tau)= \tau_0'$ if $\tau<\tau_1$;
\eiv and
$\delta_{\B}(\tau)= \tau_n'$ if $\tau_n<\tau$;
\ev
and for $c\in \dom$ if $c \in \adomExt$ then $\delta_{\B}(c)=c$;
otherwise $\delta_{\B}(c)=a$.
If $\B$ has no comparisons then $\delta_{\B}(\tau)= \tau_0'$
for each $\tau$.
%
We extend $\delta_{\B}$ to all syntactic objects containing constants,
including executions.
Now, we define $\P'$ to be as $\P$, except that 
we add conditions on each outing transition from $\start$ 
such that only instances with values from $\adomExt$ 
can traverse, and instances with the timestamps greater or equal than $\tauB$.

\begin{proposition}[Abstraction]   
      \label{prop-abstractions:of:timestamps}
Let 
$\Upsilon= \C \rightsquigarrow 
	\C_1 \rightsquigarrow \dots \rightsquigarrow\C_m$
be an execution in $\B$ that produces a set of facts $W$, and let 
$\Upsilon'=\delta_{\B}\Upsilon =\delta_{\B} \C \rightsquigarrow 
	\delta_{\B} \C_1 \rightsquigarrow \dots \rightsquigarrow \delta_{\B}\C_m$.
Further, 
let $\Upsilon''$ be an execution in $\B'$. 
Then the following holds:
\begin{enumerate}[label=\alph*)]
\item
$\Upsilon'$ is an execution in $\B'$ that produces $\delta_{\B} W$;
\item 
$Q(\D)\neq Q( \D \cup W)$ iff $Q(\D)\neq Q( \D \cup \delta_{\B} W)$;
\item
$\Upsilon''$ is an execution in $\B$. 
\end{enumerate}
\end{proposition}

In other words, 
each execution in $\B$ can be $\delta_{\B}$-abstracted and it will be an execution in~$\B'$, 
and more importantly, 
an execution in $\B$ produces a new query answer if and only if 
the $\delta_{\B}$-abstracted version produces a new query answer in $\B'$.
%

\paragraph*{Encoding into Datalog}

Since $\B'$ allows only finitely many new values in fresh instances, 
there is a bound on the maximal extensions of $\D$ that can be produced.
Moreover, since there is no bound on the number of fresh instances that can start, 
there is only a single maximal extension of $\D$, say $\D'$, that can result from $\B'$.
We now define the program $\PiPQposfr \cup \D$ whose least fixpoint is exactly this $\D'$.





First, we introduce the relations that we use in the encoding.
To record which fresh instances can reach a place $p$ in $\P$, 
we introduce for each $p$ a relation $\In_p$ with the same arity as $\In$.
That is, $\In_p(\ts)$ evaluates to true in the program iff
an instance with the input record $\In(\ts)$ can reach $p$.
As in the closed variant, we use a primed version $R'$ for each relation $R$ 
to store $R$-facts produced by the process.

Now we define the rules.
Initially, all relevant fresh instances 
(those with constants from $\adomExt$) sit at the $\start$ place.
We encode this by the \emph{introduction rule}:
$
\In_\start(X_1,\dots,X_n) \la \adomExt(X_1),\dots,\adomExt(X_n).
$
Here, with slight abuse of notation, $\adomExt$ represents a unary relation that we initially instantiate with the constants from  $\adomExt$.
Also initially, we make a primed copy of each database fact, 
that is, for each relation $R$ in $\P$ we define the \emph{copy rule}:
$R'(\tX) \la R(\tX)$. 

Then we encode instance traversals.
For every transition~$t$ that goes from a place~$q$ to a place $p$,
we introduce a \emph{traversal rule} that mimics how instances having reached $q$ 
move on to $p$, 
provided their input record satisfies the execution condition for $t$. 
Let 
$E_t = \In(\ts), R_1(\ts_1), \dots, R_l(\ts_l), G_t$
be the execution condition for $t$,
where $G_t$ comprises the comparisons.
We define the condition 
$E_t'(\ts)$ as $\In_q(\ts), R'_1(\ts_1), \dots, R'_l(\ts_l), G_t$,
obtained from $E_t$ by renaming the $\In$-atom and priming all database relations.
Then, the \emph{traversal rule} for $t$ is:
$
    \In_p(\ts) \la E'_t(\ts).
$
Here,  $E'_t(\ts)$ is defined over the primed signature
since a disabled transition may become enabled as new facts are produced.
%

Which facts are produced by traversing $t$ is captured by 
a \emph{generation rule}.
Let $W_t \colon$ $R(\tu) \la B_t(\tu)$ 
be the writing rule for~$t$, with the query
$B_t(\tu) \la \In(\ts'), R_1(\ts'_1), \dots, R_n(\ts'_n), M_t$,
where $M_t$ comprises the comparisons.
Define 
$B'_t(\ts',\tu) \la \In_q(\ts'), R'_1(\ts'_1), \dots, R'_n(\ts'_n), M_t$.
The corresponding generation rule is
$
     R'(\tu) \la E_t'(\ts) ,B_t'(\ts',\tu), \ts=\ts',
$
which combines the constraints on the instance record from $E_t$ and $W_t$.

Let $\PiPQposfr$ be the program defined above, encoding the positive fresh $\B'$ obtained from $\B$.
The program is constructed in such a way that it computes exactly the atoms
that are in the maximal extension $\D'$ of $\D$ produced by $\B'$.
%
Let $R'(\tpl v)$  be a fact. 

\begin{lemma}
\label{lemma:encoding:positive}
There is an execution in the positive fresh $\B$ producing $R(\tpl v)$  \ iff\\ 
 $\PiPQposfr \cup \D \models R'(\tpl v)$.
\end{lemma}

Let $\PiQtest$ be 
defined like  $\PiPIQtest$ in Section~\ref{sec-positive:closed},
except that there is only one rule for $Q'$,
obtained from \eqref{eqn-Q:prime:test:rule}
by replacing $R^i_j$ with $R'_j$.
Then Proposition~\ref{prop-abstractions:of:timestamps}
and Lemma~\ref{lemma:encoding:positive} imply:

\begin{theorem} 
\label{th-instability}
$Q$ is instable the positive fresh $\B$
\ iff \ 
$\PiPQposfr \cup \D \cup \PiQtest \models \instable$.
\end{theorem}

\paragraph*{Complexity}
Since $\PiPQposfr \cup \D \cup \PiQtest$ is a program with stratified negation, 
stability checking over positive fresh \DABP{s}  is in $\EXPTIME$ for process and combined complexity,
and in $\PTIME$ for data complexity \cite{Gottlob-Eiter-LP-Journal-2001}.
From Lemma~\ref{lem:positive:cyclic:complexity} we know 
that these are also lower bounds for the respective complexity measures.
Query complexity is $\PiPTWO$-complete as usual, and instance complexity is trivial 
for fresh processes.

\paragraph*{Positive \DABP{s}}
To reason about arbitrary positive \DABP{s}, we can combine the encoding
for the fresh variant ($\PiPQposfr$) from this section and the one for
the closed variant from Section~\ref{sec-positive:closed} ($\PiPIpocl$).
The main idea is that to obtain maximal extensions, 
each greedy execution sequence is augmented by also flooding the process with fresh instances.
The complexities for the full positive variant are inherited from the closed variant.



\section{Read-Only-Write-Only {\DABP}s}
\label{sec:rowo}

In general \DABP{s}, processes can perform recursive inferences by writing into relations from which they have read.
It turns out that if relations are divided into read-only and write-only, the complexity of stability reasoning drops significantly.

The main simplifications in this case are that 
\begin{inparaenum}[\itshape (i)]
\item one traversal per instance and transition suffices, since no additional fact can be produced by a second traversal;
\item instead of analyzing entire executions, 
      it is enough to record which paths an individual process instance can take
      and which facts it produces, since instances cannot influence each other.
\end{inparaenum}
As a consequence, the encoding program can be non-recursive and it is
independent of the instances in the process configuration.
A complication arises, however, since the maximal extensions 
of the original database $\D$ by the \DABP $\B$ 
are not explicitly represented by this approach.
They consist of unions of maximal extensions by each instance
and are encoded into the test query, which is part of the program.

\begin{theorem}  \label{theorem:rowo}
For every rowo \DABP $\B=\tup{\P,\I,\D}$ and query $Q$ 
one can construct a nonrecursive Datalog program $\PiPQrowo$, based on $\P$ and $Q$, 
and a database instance $\D_\I$, based on $\D$ and $\I$,
such that:
$Q$ is instable in $\B$ \ iff\ \ $\PiPQrowo \cup \D_\I \models \instable$.
\end{theorem}

From the theorem it follows that data and instance complexities are in $\ACz$, except for instance complexity in fresh variants, for
which it is constant.

\paragraph*{Process, Query and Combined Complexity}
Since CQ evaluation can be encoded into an execution condition,
this gives us $\coNP$-hardness of stability in process complexity.
We also show that it is in $\coNP$.
First we note that due to the absence of recursion, one can check in $\NP$ whether a set of 
atoms is produced by a process instance.

\begin{proposition}
 \label{prop-rowo:generability:complexity}
Let $\B$ be a singleton rowo \DABP.
One can decide in $\NP$, whether for given facts $A_1,\ldots,A_m$, 
there is an execution in $\B$ that produces $A_1,\ldots,A_m$.
\end{proposition}



Next, suppose that
$\I$, $\D$ and $Q(\tpl v)\la B_1,\ldots,B_m$ are a fixed 
instance part, database and query.
Given a process model $\P$, 
we want to check that $Q$ is instable in $\B_\P=\tup{\P,\I,\D}$.
Making use of the abstraction principle for fresh constants, 
we can guess in polynomial time 
an instantiation $B'_1,\ldots,B'_n$ of the body of $Q$ that 
returns an answer not in $Q(\D)$.
Then we verify that $B'_1,\ldots,B'_n$ are produced by $\B_\P$.
Such a verification is possible in $\NP$ according to 
Proposition~\ref{prop-rowo:generability:complexity}.
We guess a partition of the set of facts $B'_1,\ldots,B'_n$,
guess one instance, possibly fresh, for each component set of the partition,
and verify that the component set is produced by the instance.
Since all verification steps were in $\NP$, the whole check is in $\NP$.


Query complexity is $\PiPTWO$-complete for the same reasons as in the general variant,
and one can show that this is also the upper-bound for the combined complexity.


\section{Related Work}
\label{ses:related:work}

Traditional approaches for business process modeling focus on the set of activities to be performed and the flow of their execution. 
These approaches are known as \emph{activity-centric}. 
A different perspective, mainly investigated in the context of databases, consists in identifying the set of data (entities) to be represented and describes processes in terms of their possible evolutions. 
These approaches are known as \emph{data-centric}.

%
In the context of activity-centric processes, Petri Nets (PNs) have been used for the representation, validation and verification of formal properties, such as absence of deadlock, boundedness and reachability
\cite{Aalst-Verification-97, Aalst-Patterns-03}.
In PNs and their variants, a token carries a limited amount of information, which can be represented by associating to the token a set of variables, like in colored PNs 
\cite{jensen-colouredPN-2009}. 
No database is considered in PNs.
%


Among data-centric approaches,
\emph{Transducers}~\cite{Abiteboul-Transducers-PODS-98,Spielmann}
were among the first formalisms ascribing a central role to the
data and how they are manipulated. 
%
These have been extended 
to \emph{data driven web systems}~\cite{Deutsch-WebData-PODS-2004}
to model
the interaction of a user with a web site, 
which are then extended in
\cite{ltl:verification:hull:vianu-icdt2009-decidable}.
These frameworks
express 
insertion and deletion rules using FO formulas. 
The authors verify properties
expressed as FO variants of LTL, CTL and CTL* temporal formulas.
The verification of these formulas results to be
undecidable in the general case. 
Decidability is obtained under certain restrictions
on the input, yielding to 
\EXPSPACE complexity for checking LTL formulas
and \coNEXPTIME and \EXPSPACE for 
CTL and CTL* resp., in the propositional case.

Data-Centric Dynamic Systems (DCDSs) 
\cite{Calvanese-DCDS-PODS-2013}
describe processes in terms of guarded FO rules that evolve the database.
%
%
The authors study the verification of temporal properties
expressed in variants of $\mu$-calculus (that subsumes CTL*-FO).
%
%
They identify several undecidable classes 
and isolate decidable variants 
by 
assuming a bound on the size of the database at each step or
 a bound on the number of constants at each run.
In these cases verification is
\EXPTIME-complete in data complexity.

Overall, both frameworks are more general than \DABP{s}, since deletions and updates of facts are also allowed. 
This is done by rebuilding the database after each execution step.
Further, our stability problem can be encoded as FO-CTL formula.
However, 
our decidability results for positive \DABP{s}
are not captured by the decidable fragments of those approaches.
In addition, the authors of the work above investigate the borders of decidability, while we focus on a simpler problem and study the sources of complexity.
%
%
Concerning the process representation, 
both approaches adopt a rule-based specification.
This makes the control flow more difficult to grasp,
in contrast to activity-centric approaches where 
the control flow has an explicit representation.

\emph{Artifact-centric} approaches
\cite{Hull-Artifact-Centric-Survey} 
use artifacts to model business relevant entities.
%
%
In 
\cite{artifact-formal-analysis,
artifct-static-analysis,
gerede-su-icsoc2007}
the authors investigate the verification of properties 
of artifact-based processes
such as reachability,  temporal constraints, and  
the existence of dead-end paths.
%
%
However, none of these approaches explicitly models an underlying database.
Also,
the authors focus on finding suitable restrictions 
to achieve decidability, without a fine-grained complexity analysis as in our case.

Approaches in \cite{Milo-Provenance} and \cite{MILO-BP-QL},
investigate the challenge of combining processes and data, however, 
focusing on the problem of data provenance and of querying the process structure.

In~\cite{Elkan-QueryUpdate-PODS-90,LevySagiv-updates-1993-VLDB} the authors
study the problem of determining whether a query over views is independent 
from a set of updates over the database. 
The authors do not consider a database instance nor a process.
Decidability in rowo \DABP{s} can be seen as a
special case of those. 

In summary, 
our approach to process modeling is closer to the activity-centric one
but we model manipulation of data like in the data-centric
approaches.
Also, having process instances and \DABP{s} restrictions 
gives finer granularity
compared to data-centric approaches.


\section{Discussion and Conclusion}
\label{sec:discussion:conlcusion}

\paragraph*{Discussion}
An interesting question is how complex stability becomes
if \DABP{s} are not monotonic, i.e.,   
if updates or deletions are allowed.
In particular, for positive \DABP{s} we can show the following.
In acyclic positive closed \DABP{s} updates and deletions
can be modeled using negation in the rules,
thus stability stays $\PSPACE$-complete.
For the cyclic positive closed variant, allowing updates or deletions is more powerful
than allowing negation, and stability 
jumps to $\EXPSPACE$-completeness.
For positive \DABP{s} with updates or deletions 
stability is undecidable. 

In case the initial database is not known,
our techniques can be still applied since an arbitrary database can be produced by fresh instances starting from an empty database.

\paragraph*{Contributions}
Reasoning about data and processes can be relevant in decision support
to understand how processes affect query answers.
(1) To model processes that manipulate data we adopt an explicit representation of the control flow as in standard BP languages (e.g., BPMN).
We specify how data is manipulated as annotations on top of the control flow.
(2) Our reasoning on stability can be offered as a reasoning service on top of the query answering that reports on the reliability of an answer.
Ideally, reasoning on stability should not bring a significant overhead on query answering in practical scenarios.
Existing work on processes and data~\cite{Calvanese-DCDS-PODS-2013}
shows that verification of general temporal properties is typically 
intractable already measured in the size of the data.
(3) In order to identify tractable cases and sources of complexity we investigated different variants of our problem, by considering negation in conditions, cyclic executions, read access to written data, presence of pending process instances, and the possibility to start fresh process instances.
%
(4) Our aim is to deploy reasoning on stability to existing query answering platforms such as SQL and ASP~\cite{DLV-Leone-2006}.
For this reason we established different encodings into suitable variants of Datalog, that are needed to capture the different characteristics of the problem. For each of them we showed that our encoding is optimal.
In contrast to existing approaches, which rely on model checking to verify properties, in our work we rely on established database query languages.



\paragraph*{Open Questions}
In our present framework we cannot yet
model process instances with activities that are running in parallel.
Currently, we are able to deal with it only in case instances do not interact (like in rowo).
Also, we do not know yet how to reason about expressive queries, such as conjunctive queries with negated atoms,
and first-order queries.
%
%
From an application point of view,
stability of aggregate queries and  aggregates in the process rules are relevant.
A further question is how to  quantify instability, 
that is, in case a query is not stable, how to compute the minimal and  maximal number of possible new answers.

\subparagraph*{Acknowledgements}
This work was partially supported by the research projects
MAGIC, funded by the province of Bozen-Bolzano, 
and 
CANDy and PARCIS, funded by the Free University 
of Bozen-Bolzano.

\newpage

 \bibliography{bib/ref} 

\begin{thebibliography}{10}

\bibitem{Abiteboul-Transducers-PODS-98}
S.~Abiteboul, V.~Vianu, B.S. Fordham, and Y.~Yesha.
\newblock {Relational Transducers for Electronic Commerce}.
\newblock In {\em PODS}, pages 179--187, 1998.

\bibitem{Fagin-3ColorExtension}
M.~Ajtai, R.~Fagin, and L.J. Stockmeyer.
\newblock {The Closure of Monadic {NP}}.
\newblock {\em J. Comput. Syst. Sci.}, 60(3):660--716, 2000.

\bibitem{Milo-Provenance}
Y.~Amsterdamer, S.~B. Davidson, D.~Deutch, T.~Milo, J.~Stoyanovich, and
  V.~Tannen.
\newblock {Putting Lipstick on Pig: Enabling Database-Style Workflow
  Provenance}.
\newblock {\em {PVLDB}}, 5(4):346--357, 2011.

\bibitem{Calvanese-DCDS-PODS-2013}
B.~{Bagheri Hariri}, D.~Calvanese, G.~{De Giacomo}, A.~Deutsch, and M.~Montali.
\newblock {Verification of Relational Data-Centric Dynamic Systems with
  External Services}.
\newblock In {\em PODS}, pages 163--174, 2013.

\bibitem{MILO-BP-QL}
C.~Beeri, A.~Eyal, S.~Kamenkovich, and T.~Milo.
\newblock {Querying Business Processes}.
\newblock In {\em VLDB}, pages 343--354, 2006.

\bibitem{artifact-formal-analysis}
K.~Bhattacharya, C.~E. Gerede, R.~Hull, R.~Liu, and J.~Su.
\newblock {Towards Formal Analysis of Artifact-Centric Business Process
  Models}.
\newblock In {\em BPM}, pages 288--304, 2007.

\bibitem{bonita}
Bonitasoft.
\newblock {Bonita BPM}.
\newblock \url{www.bonitasoft.com}.
\newblock Accessed: 2015-12-16.

\bibitem{Wenfei-Consistency-VLDB-07}
G.~Cong, W.~Fan, F.~Geerts, X.~Jia, and S.~Ma.
\newblock {Improving Data Quality: Consistency and Accuracy}.
\newblock In {\em VLDB}, pages 315--326, 2007.

\bibitem{Gottlob-Eiter-LP-Journal-2001}
E.~Dantsin, T.~Eiter, G.~Gottlob, and A.~Voronkov.
\newblock {Complexity and Expressive Power of Logic Programming}.
\newblock {\em ACM Comput. Surv.}, 33(3):374--425, 2001.

\bibitem{ltl:verification:hull:vianu-icdt2009-decidable}
A.~Deutsch, R.~Hull, F.~Patrizi, and V.~Vianu.
\newblock {Automatic Verification of Data-Centric Business Processes}.
\newblock In {\em ICDT}, pages 252--267, 2009.

\bibitem{Deutsch-WebData-PODS-2004}
A.~Deutsch, L.~Sui, and V.~Vianu.
\newblock {Specification and Verification of Data-Driven Web Services}.
\newblock In {\em PODS}, pages 71--82, 2004.

\bibitem{Elkan-QueryUpdate-PODS-90}
C.~Elkan.
\newblock {Independence of Logic Database Queries and Updates}.
\newblock In {\em PODS}, pages 154--160, 1990.

\bibitem{Fan:Et:Al-Currency-TODS}
W.~Fan, F.~Geerts, and J.~Wijsen.
\newblock {Determining the Currency of Data}.
\newblock {\em ACM Trans. Database Syst.}, 37(4):25, 2012.

\bibitem{artifct-static-analysis}
C.~E. Gerede, K.~Bhattacharya, and J.~Su.
\newblock {Static Analysis of Business Artifact-Centric Operational Models}.
\newblock In {\em SOCA}, pages 133--140, 2007.

\bibitem{gerede-su-icsoc2007}
C.~E. Gerede and J.~Jianwen~Su.
\newblock {Specification and Verification of Artifact Behaviors in Business
  Process Models}.
\newblock In {\em {ICSOC}}, pages 181--192, 2007.

\bibitem{bizartifact}
F.~T. Heath, D.~Boaz, M.~Gupta, R.~Vacul\'{\i}n, Y.~Sun, R.~Hull, and
  L.~Limonad.
\newblock {Barcelona: A Design and Runtime Environment for Declarative
  Artifact-Centric BPM}.
\newblock In {\em ICSOC}, pages 705--709, 2013.

\bibitem{Hull-Artifact-Centric-Survey}
R.~Hull.
\newblock {Artifact-Centric Business Process Models: Brief Survey of Research
  Results and Challenges}.
\newblock In {\em {OTM}}, pages 1152--1163, 2008.

\bibitem{jensen-colouredPN-2009}
K.~Jensen and L.M. Kristensen.
\newblock {\em Coloured Petri Nets: Modelling and Validation of Concurrent
  Systems}.
\newblock Springer, 2009.

\bibitem{DLV-Leone-2006}
N.~Leone, G.~Pfeifer, W.~Faber, T.~Eiter, G.~Gottlob, S.~Perri, and
  F.~Scarcello.
\newblock {The {DLV} System for Knowledge Representation and Reasoning}.
\newblock {\em {ACM} Trans. Comput. Log.}, 7(3):499--562, 2006.

\bibitem{LevySagiv-updates-1993-VLDB}
A.Y. Levy and Y.~Sagiv.
\newblock {Queries Independent of Updates}.
\newblock In {\em VLDB}, pages 171--181, 1993.

\bibitem{AMW}
E.~Marengo, W.~Nutt, and O.~Savkovi\'c.
\newblock {Towards a Theory of Query Stability in Business Processes}.
\newblock In {\em AMW}, volume 1189 of {\em {CEUR} Workshop Proceedings}, 2014.

\bibitem{bpmn}
{Object Management Group}.
\newblock {\em {Business Process Model and Notation 2.0 (BPMN)}}, Jan 2011.

\bibitem{Razniewski-Completeness-VLDB-11}
S.~Razniewski and W.~Nutt.
\newblock {Completeness of Queries over Incomplete Databases}.
\newblock {\em PVLDB}, 4(11):749--760, 2011.

\bibitem{DABPS-KRDB:Report}
O.~Savkovi\'c, E.~Marengo, and W.~Nutt.
\newblock {Query Stability in Data-aware Business Processes}.
\newblock Technical Report KRDB15-1, KRDB Research Center, Free Univ.
  Bozen-Bolzano, 2015.
\newblock \url{http://www.inf.unibz.it/krdb/pub/tech-rep.php}.

\bibitem{Spielmann}
M.~Spielmann.
\newblock {Verification of Relational Transducers for Electronic Commerce}.
\newblock In {\em PODS}, pages 92--103. ACM, 2000.

\bibitem{Aalst-Verification-97}
W.M.P. van~der Aalst.
\newblock {Verification of Workflow Nets}.
\newblock In {\em ICATPN}, pages 407--426, 1997.

\bibitem{Aalst-Patterns-03}
W.M.P. van~der Aalst, A.H.M. ter Hofstede, B.~Kiepuszewski, and A.P. Barros.
\newblock {Workflow Patterns}.
\newblock {\em Distributed and Parallel Databases}, 14(1):5--51, 2003.

\end{thebibliography}

\newpage

\appendix

\def\thepart{}

\part{Appendices}
\parttoc 


\section{Example}
\label{section:appendix:example}

As an illustration of the concepts in our formalism,
we provide an example about student enrollment at a university.

\subsection{Scenario: Student Registration}
One year in November, the student office distributes a report
on the numbers of new student registrations for the offered programs.
When comparing the numbers with those of the previous years, 
the Master in Computer Science (\emph{mscCS}) shows a decrease, 
in contrast with the Master in Economics (\emph{mscEco}), 
which has registered a substantial increase.
An analysis task force at university level cannot identify a plausible cause.
Eventually, a secretary discovers 
that the reason is a complication in the registration process, 
which foresees two routes to registration: 
a regular one and a second one via international federated study programs 
to which some programs, like the \emph{mscCS}, are affiliated.
Due to different deadlines, regular registration has been concluded in November
while registration for students from federated programs has not. 
Since the \emph{mscCS} is affiliated to some federated programs, but the \emph{mscEco} is not,
the query asking for all \emph{mscEco} students was stable in November,
while the query for all \emph{mscCS} students was not and returned too 
low a number.

\subsection{\DABP Representation of the Scenario}

Table~\ref{tb:registration} shows
the student registration process $\Breg=\langle \Preg, \Creg \rangle$.
Part~\emph{(a)} reports the process net and
Part~\emph{(b)} the execution conditions and writing rules.

A process instance starts when a student submits an online request,
providing as input the student name $S$ and the course name $C$. 
Automatically, the system attaches a time stamp $\T$ to the request.
The application is then represented as an $\In$-record $\In(S,C, \T\,)$.

The procedure distinguishes between applications
to \emph{international} courses, 
which are part of programs involving universities from different countries
and where an international commission decides whom to admit, 
and to \emph{regular} programs, where the university itself evaluates the applications.
Accordingly, a process first checks for the type of application.
The transition \qact{is intl.~app.} can only be traversed,
if the execution condition $\In(S,C,\T\,), \studyplan(C,{\transf},P)$ succeeds,
which is the case when the course of the application
is stored in the relation $\studyplan$ and associated to a program with type $\text{\transf}$.
Subsequently, the process checks if the student has already been admitted \act{is admitted}.
If so, it pursues the upper branch of the net.
If not, it checks if the course is also open to regular
applications \act{isn't admitted}. 
Similarly, the execution condition on \qact{is reg.~app.} ensures that the course
is associated to a program of type $\text{\ord}$, but not of type  $\text{\transf}$.
Then, applications for regular courses follow the bottom branch.

Deadlines give rise to conditions on the application timestamp $\T$.
While applications are accepted starting from \first\ Sep, 
the deadline for regular courses is 3\first\ Oct,
and for international courses it is 3\first\ Dec.
Candidates who applied until 30$\nth$ Sep can pre-enroll, 
that is, register provisionally.
After that date, admitted candidates have to register directly.

Provisional registration gives students the possibility
\ei to enroll conditionally and complete an application not fully complete, and
\eii to confirm or withdraw the registration before being formally enrolled.
Modeling the completion of incomplete applications leads to cycles in the net,
while non-determinism, e.g.\ due to human intervention or interaction with other systems, 
is modeled by labeling the transitions emanating from a place 
(like \qact{acad.~check} or \qact{stud.~decis.}) with non-exclusive execution conditions.

Some transitions are labeled with a writing rule. 
When \qact{pre-enrol cond.} is traversed, 
the rule $\tlbl{Conditional}(S,C) \la \In(S,C,\T\,)$
records that the application is conditionally accepted
by writing a fact into the relation $\tlbl{Conditional}$.
This relation, on the other hand, is read by the execution condition of the transition \qact{complete app.} 


\newcommand{\fsize}{\scriptsize}

\begin{table*}[t!]

\begin{tabular}{ll}
%
%
\begin{minipage}{0.54\linewidth}
\mbox{}

\vspace{5ex}
%
%
\begin{minipage}{1\linewidth}
\begin{center}
\includegraphics[width=1.05\linewidth]{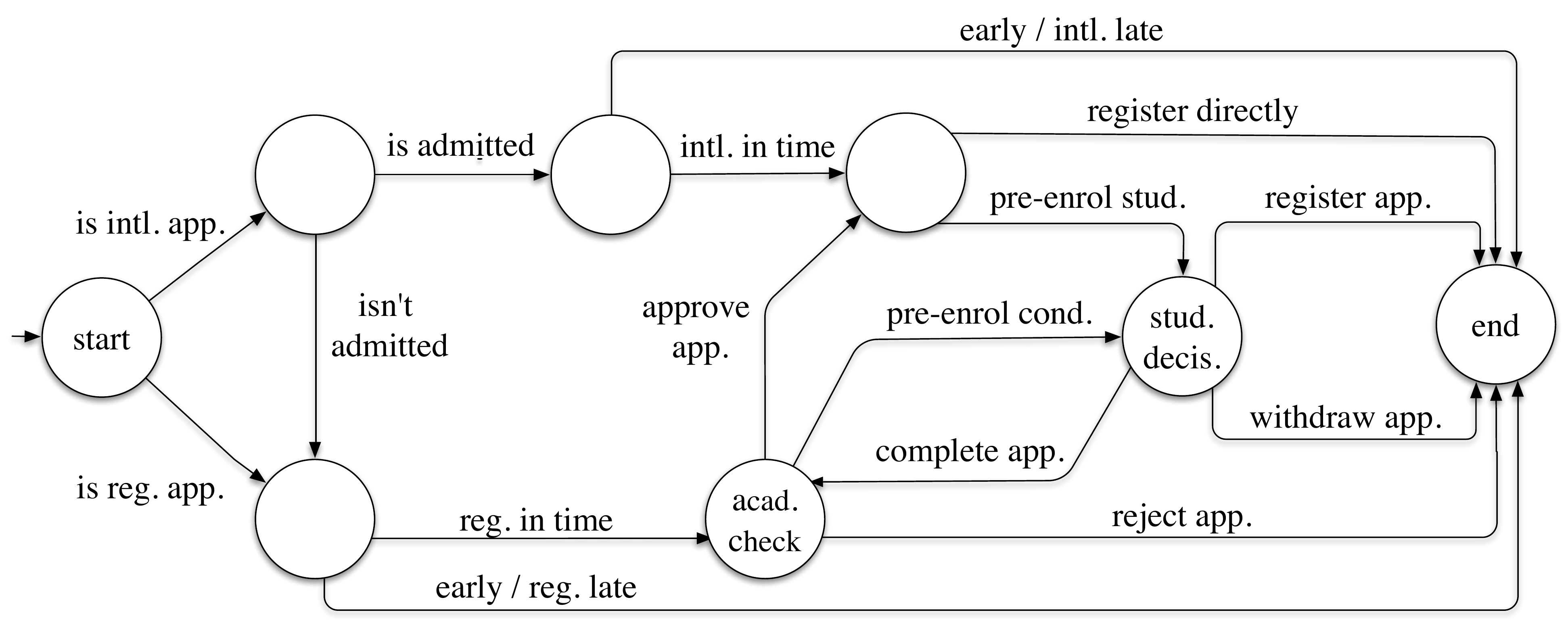}
{\small\emph{\textbf{(a)} Process Net in $\Preg$}}\\[1.5em]
\end{center}
\end{minipage}\\[7.8ex]
%
%
%
\begin{minipage}{0.66\linewidth}
\begin{scriptsize}
\begin{tabular}{@{}l@{\hspace{2ex}}l@{}}
\fsize  \textit{Transition} & 
\fsize  \textit{Execution Condition $(\LE{})$} \\ 
\midrule
\fsize 	is intl.~app. &
\fsize	$\In(S,C,\T\,), \studyplan(C,\text{\transf},P)$\\[0.5ex]
\fsize	is admitted &
\fsize	$\In(S,C,\T\,), \tlbl{AdmittedIntl}(S,C)$\\[0.5ex]
\fsize	isn't admitted &
\fsize  $\In(S,C,\T\,),$ \parbox[t]{10em}{$\neg \tlbl{AdmittedIntl}(S,C),$ \\
                                          $\tlbl{\studyplan}(C,\text{\ord},P)$}\\[3.3ex]
\fsize	is reg.~app. &
\fsize	$\In(S,C,\T\,),$ \parbox[t]{11em}{$\tlbl{\studyplan}(C,\text{\ord},P),$\\
                                          $\neg \tlbl{\studyplan}(C,\text{\transf},P)$}\\[3.3ex]
\fsize	reg.~in time &
\fsize	{$\In(S,C,\T\,), \mbox{\first\ Sep} \leq \tstamp \leq \mbox{3\first\ Oct}$} \\[0.5ex]
\fsize	intl.~in time  &
\fsize	{$\In(S,C,\T\,), \mbox{\first\ Sep} \leq \tstamp \leq  \mbox{3\first\ Dec}$} \\[0.5ex]
\fsize	early  &
\fsize	{$\In(S,C,\T\,), \tstamp < \mbox{\first\ Sep}$}\\[0.5ex]
\fsize	reg.~late  &
\fsize	{$\In(S,C,\T\,), \tstamp > \mbox{3\first\ Oct}$}\\[0.5ex]
\fsize	intl.~late  &
\fsize	{$\In(S,C,\T\,), \tstamp > \mbox{3\first\ Dec}$}\\[0.5ex]
\fsize	pre-enrol stud. &
\fsize	{$\In(S,C,\T\,), \tstamp \leq \mbox{\dateconditional}$}\\[0.5ex]
\fsize	pre-enrol cond. &
\fsize	{$\In(S,C,\T\,), \tstamp \leq \mbox{\dateconditional}$}\\[0.5ex]
\fsize	register directly &
\fsize	{$\In(S,C,\T\,), \tstamp > \mbox{\dateconditional}$}\\[0.5ex]
\fsize	complete app. &
\fsize	{$\In(S,C,\T\,), \tlbl{\conditional}(S,C)$}\\[0.5ex]
\fsize	register app. &
\fsize	{$\In(S,C,\T\,), \tlbl{\accepted}(S,C)$}\\[0.5ex]
\multicolumn{2}{@{}l}{\fsize withdraw app.~= approve app.~= reject app.$\colon \true$}\\[0.5ex]

\bottomrule
\end{tabular}
\end{scriptsize}
\end{minipage}
\end{minipage}
& 
%
%
\begin{minipage}{0.99\linewidth}
\begin{minipage}{0.40\linewidth}
\begin{center}
\begin{scriptsize}
\begin{tabular}{@{}lll@{}}
 \multicolumn{3}{c}
 {\fsize \studyplan} \\
\toprule 
 \fsize\course & \fsize\textit{type}  & \fsize\program \\
\midrule
 \fsize\emcl   & \fsize\emph{\transf} & \fsize\msccs\\
 \fsize\emcl   & \fsize\emph{\ord}    & \fsize\msccs\\
 \fsize db     & \fsize\emph{\ord}    & \fsize\msccs\\
 \fsize econ   & \fsize\emph{\ord}    & \fsize\msceco\\
\bottomrule
\end{tabular}
\end{scriptsize}
\end{center}
\end{minipage}
\\[2ex]
%
\begin{minipage}{0.5\linewidth}
\begin{scriptsize}
\begin{tabular}{@{}ll@{}}
 \multicolumn{2}{c}
 {\fsize\textit{AdmittedIntl}} \\
\toprule
 \fsize\textit{student} &  \fsize\course\\
\midrule
 \fsize bob  & \fsize\emcl \\
 \fsize mary & \fsize\emcl \\
\bottomrule    
\end{tabular}
\end{scriptsize}
\vspace{2.0ex}
\begin{scriptsize}
\begin{tabular}{@{}ll@{}}
 \multicolumn{2}{c}
     {\fsize \accepted} \\
\toprule
 \fsize \textit{student} &  \fsize \course \\
\midrule
 \fsize bob   & \fsize \emcl \\
 \fsize alice & \fsize \econ \\
\bottomrule
\end{tabular}
\end{scriptsize}
\end{minipage}
\\
%
\begin{minipage}{0.5\linewidth}
\begin{scriptsize}
\begin{tabular}{@{}ll@{}}
 \multicolumn{2}{c}
 {\fsize\conditional} \\
\toprule
 \fsize\textit{student} & \fsize\course \\
\midrule
 \fsize paul & \fsize\econ\\
\bottomrule
\vspace*{0.12ex}
\end{tabular}
\end{scriptsize}
\hspace{1.2em}
\begin{scriptsize}
\begin{tabular}{@{}ll@{}}
 \multicolumn{2}{c}
 {\fsize\registered} \\
\toprule
  \fsize\textit{student} &  \fsize\course \\
\midrule
  \fsize bob &   \fsize\emcl\\
  \fsize alice &  \fsize\econ \\
\bottomrule
\end{tabular}
\end{scriptsize}
\end{minipage}
\\[2ex]
\emph{\small\hspace*{2em} \textbf{(c)} Database Instance $\Dreg$}
\\[2ex]
%
\begin{minipage}{0.5\linewidth}
\begin{scriptsize}
\begin{tabular}{ccc}
 \multicolumn{3}{c}
 {\fsize\textit{Instances and Mappings}} \\
\toprule
 \fsize $\textit{id}$ & \fsize$\In$-record & \fsize place \\
\midrule
 \fsize$o_1$ & \fsize (bob, \emcl, $5\nth$ Sep) & \fsize\terminate\\ 
 \fsize$o_2$ & \fsize (alice, \econ, $20\nth$ Sep) & \fsize\terminate\\ 
 \fsize$o_3$ & \fsize (paul, \econ, $24\nth$ Sep) & \fsize\terminate\\ 
 \fsize$o_4$ & \fsize (john, db, $4\nth$ Nov) & \fsize\start\\ 
\bottomrule
\end{tabular}
\end{scriptsize}
\end{minipage}
\\[2ex]
\emph{\small\hspace*{2em} \textbf{(d)} Process Instances $\Ireg$}
\\[3ex]
%
\mbox{}\hspace{-3.0em}%
\begin{minipage}{0.25\linewidth}
\begin{scriptsize}
\begin{tabular}{@{}lr@{}}
\fsize\textit{Transition} & \fsize{\textit{Writing Rule $(\LW{})$}}\\ 
\midrule
\fsize	register directly &
	\fsize{$ \tlbl{\registered}(S,C) \la \In(S,C,\T\,) $}\\[0.5ex]
\fsize	pre-enrol stud.&
	\fsize{$\tlbl{\accepted}(S,C) \la \In(S,C,\T\,)$}\\[0.5ex]
\fsize	pre-enrol cond. &
	\fsize{$\tlbl{Conditional}(S,C) \la \In(S,C,\T\,)$}\\[0.5ex]
\fsize	register app. &
	\fsize{$\tlbl{\registered}(S,C) \la \In(S,C,\T\,)$}\\[0.5ex]
\bottomrule
\end{tabular}
\end{scriptsize}
\end{minipage}
\end{minipage}
\end{tabular}
\vspace{1.3ex}
\hspace*{10em} \emph{\small\textbf{(b)} Execution Conditions and Writing Rules in $\Preg$}

\vspace{1.3ex}

\caption{\DABP Representation of the Student Registration Process.}
\label{tb:registration}
\end{table*}


Table~\ref{tb:registration}\emph{(c)} shows 
a database instance $\Dreg$ for our running example.
Courses offered are stored in the relation \emph{StudyPlan}, 
together with their type (\transf\ or \ord) 
and the program they are associated with.
The remaining tables store information about the students.
Table~\ref{tb:registration}\emph{(d)} 
reports the running process instances $\Ireg$ in the form of a relation. 

\subsection{Stability of the Example Queries}

Consider the queries 
\begin{align*}
& Q_\text{eco}(S) \leftarrow \textit{Registered}(S,C),
	\studyplan(C, T, \textit{mscEco}), \\
& Q_\text{cs}(S) \leftarrow \textit{Registered}(S,C), 
	\studyplan(C, T, \textit{mscCS}),
\end{align*}
which ask for the students registered for the master in CS, 
and the master in Economics, respectively.
We analyze their stability over different periods,
specified in Table~\ref{table:stability-intervals}.
For each period from $\tau_1$ to  $\tau_2$, 
we ask if the query is stable until $\tau_2$ in a variant of $\Breg$
where 
\ei the current date is in the interval and 
\eii there are no running applications with a start date later than~$\tau_2$
(and also no data in the tables about the students having submitted one).

\begin{table}[!ht]
	\begin{center}
		\begin{tabular}{@{} c c c c c @{}}
		\toprule
		\emph{Period} 		   & 
		 {$<$ \first\ Sep} &
		 {\first\ Sep--3\first\ Oct} &
		 {\first\ Nov--3\first\ Dec} &
		 {$>$ 3\first\ Dec} \\ \midrule 
		 $Q_\text{eco}$ 		   & 
		 stable 			   &
		 instable 		   	   &
		 stable				   &
		 stable				   \\ 
		 $Q_\text{cs}$ 		   & 
		 stable 			   &
		 instable 		  	   &
		 instable 	                   &
		 stable				   \\ \bottomrule
		\end{tabular}
	\end{center}
\caption{Stability of the queries $Q_\text{eco}$ and $Q_\text{cs}$ for different intervals.}
      \label{table:stability-intervals}
      \vspace*{-0.5em}
\end{table}%

During the period before \first\ Sep,
neither program allows registrations to proceed
and thus both queries are stable until this date.
For the period {\first\ Sep--3\first\ Nov},
the programs allow for new registrations and 
both queries are instable.
If the current time is within the period \first\ Nov--3\first\ Dec
and there are no pending applications,
$Q_\text{eco}$ is stable 
because the program \emph{mscEco} is not affiliated with any international (\transf) course and 
the deadline for the regular programs has passed. 
However, \emph{mscCS} has an affiliated course to which student \emph{Mary} is admitted.
She is not registered yet and potentially could submit an application 
before the 3\first\ Dec, which would be accepted.
Thus, $Q_\text{cs}$ is not stable for this period. 
If all the admitted students had already been registered,
the query would be stable, 
since no new registration would be possible. 
The query would also be stable in the case 
the process is closed for new instances to start
(e.g., because the limit on registered students has been reached).
In this case, only running instances would be allowed to finish their execution. 
Thus, candidate \emph{Mary} would not be able to register even though she is admitted. 
If the current time is after $3$\first\ Dec, both queries are stable regardless 
whether the process is closed or not because all the registration deadlines have expired.

\subsection{Variants of the Example Process}

%
%
The model of $\Breg$ is general, since the
relations \emph{Pre-enrolled} and \emph{Conditional} are both read
and written;
the rules are normal, though only with negation on database relations
that are not updated;
the net is clearly cyclic.
We can imagine that at the beginning
of the registration period the process starts with a fresh configuration
(i.e., no running applications).
The case of arbitrary configurations includes situations that 
arise as exceptions in the registration process and
cannot evolve from a fresh configuration.
For instance, a regular application received after the deadline for a valid
reason may be placed by a secretary at a certain place in the process 
that it would not be able to reach from the \start\ place.
Our example is not closed. 
If after the last deadline (3\first\ \decd)
the web form for submitting new applications will be no
more available the process will run under closed executions.

Note that, in our running example, negation in the conditions appears only
on the database relations that are not updated by the process.
For this case we can still apply the encoding from this section and obtain 
a semipositive Datalog program.
%


\section{Closed \DABP{s}}
\label{subsubsec:normal:cyclic:arbitrary:extended}




\subsection{Proof of Theorem~\ref{theorem:sms}}

The encoding program consists of the following rules:
(a)~\emph{Ordering rules} that generate a linear order of  
size  $m k r c^a$ (see Section~\ref{sec-positive:closed}) that we use to enumerate all executions steps;
(b)~\emph{Selection rules} that for each execution step non-deterministically 
selects one instance and transition meaning that the selected instances 
traverses the selected transition at that step;
(c)~\emph{Control rules} that discard cases where guesses execution sequence do not correspond to a valid execution in the process;
(d)~\emph{Generation rules} that generate facts produced by a valid execution;
(e)~\emph{Testing rules} that test if any of the guessed executions yield a new query answer.





\paragraph*{Generating Exponentially Big Linear Order}



Assume we are given an \DABP $\BPID$  
possible with cycles and negation in the rules.
As we discussed, 
to check stability in $\B$ cyclic \DABP{s} it is sufficient to consider
executions that have up to
$
	m k r c^a 
$
executions steps,
where $	m, k, r, c$ and $a$ are parameters of $\B$ as defined 
in Section~\ref{sec-positive:closed}.
%
%

%
%
%
%
%
%
%
We introduce a Datalog program that
generates a liner order of size $m k r c^a$ 
starting from a much smaller order (exponentially smaller).
To define a small order we introduce 
a set of constants, called \emph{digits},
$
\Dig_\B =\set{d_1,\dots,d_l}
$
of size $l$ an we establish one linear order $<$ on $\Dig_\B$: 
$
d_1<d_2<\dots<d_l.
$
Assume that $\Dig_\B^g$ is the Cartesian power of $\Dig_\B$ of size $g$.
That is, each tuple $\tk$ from $\Dig_\B^g$ is of the form
$
 \tk = \tup{d_{i_1},\dots,d_{i_g}}, 
$
for $d_{i_1},\dots,d_{i_g} \in \Dig_\B$.
We define $<^g$ as the lexicographical order 
on the tuples from $\Dig_\B^g$.
Here $l$ and $g$ are selected such that 
\ei
$ l= kc$ and thus depends on $\P, \I$ and $D$, and 
\eii
$g > m+r+a$ 
where $m$, $r$, and $a$ are the parameters that depend only on $\P$.
Then, it is not hard to check that it holds
\[
l^g \ge m k r c^a.
\]
In other words, 
linear order on $<^g$ is sufficient to enumerate all executions steps.

Adopting the idea in \cite{Gottlob-Eiter-LP-Journal-2001},
we define a positive Datalog program 
that generates $<^g$.
In particular, we want generate a relation $\Succ$ that stores
the immediate successor in the order.
The order $<^g$ is generated based on the orders $<^i$
on $\Dig_\B^i$ for $i<g$. 
For that and to count the execution steps 
we introduce the the following relations.

\textbf{Notation}
We introduce:
\ei	
$\Digit$ -- a unary relation such that $\Digit(d)$ is true iff $d \in \Dig_\B$;
\ei
$\First^i$ -- 
		 an auxiliary relation of arity $i$ such that 
		$\First^i(\td)$ is true iff
  		$\td$ is the first element of the linear order $<^i$;
\eii 
$\Last^i$ -- 
	 an auxiliary relation of arity $i$ such that 
		$\Last^i(\td)$ is true iff
  		$\td$ is the last element of the linear order $<^i$;
\eii 
 $\Key$ -- a $g$-ary relation such that $\Key(\tk)$ is true iff $\tk$ is a tuple from $\Dig_\B^g$ that
  corresponds to some execution step, that is any tuple from $\Dig_\B^g$ except for the first one in the order $<^g$;
 \eiii
 	$\Succ^i$ -- a relation of arity $2i$ such that
 	$Succ^i(\td_1,\td_2)$  is true 
 	iff $\td_2$ is the immediate successor of $\td_1$ in the order $<^i$.

\myparagraph{Ordering Rules}
To generate $\Succ^i$, $\First^i$, and $\Last^i$ we introduce the following rules:
\begin{align*}
	& \Succ^{i+1}(Z,\tX;Z,\tY) \la \Digit(Z), \Succ^i(\tX;\tY) \\
	& \Succ^{i+1}(Z_1,\tX;Z_2,\tY) \la \Succ^1(Z_1;Z_2), \First^i(\tX), \Last^i(\tY) \\
	&\First^{i+1}(X_1,\tX) \la \First^1(X_1), \First^i(\tX) \\
	&\Last^{i+1}(Y_1,\tY) \la \Last^1(Y_1), \Last^i(\tY).
\end{align*}

\noindent
Then, we populate relation $\Key$ with the following rule:
\begin{align*}
	\Key(K_1,\dots,K_g) \la \Digit(K_1),\dots,\Digit(K_g), \neg \First(K_1,\dots,K_g)
\end{align*}

\noindent
In the following we use $\Succ$ for $\Succ^g$,
$\First$ for $\First^g$, and $\Last$ for $\Last^g$.

We denote the above program as
$
\PiSucc_\P \cup \D_\B
$.
Here,  $\PiSucc_\P$ is a program that is polynomial in the size of $\P$, and 
$\D_\B$  is a database instance that contains
facts for relations $\Digit$, $\First^1$, $\Last^1$ and 
$\Succ^1$, and thus it is that is polynomial in the size of $\B$.
Then it holds:
%
\begin{lemma} 
\label{lemma:succ:order}
Let $\tk$, $\tk_1$ and $\tk_2$ be tuples of size $g$, then:
\begin{itemize}
	\item[\ei] $\PiSucc_\P \cup \D_\B \models \Succ(\tk_1,\tk_2)$ 
	\ iff \
	$\tk_2$ is the successor of $\tk_1$. 
	\item[\ei] 
	$\PiSucc_\P \cup \D_\B \models \Key(\tk)$ 
	\ iff \
	$\tk \in \Dig^g$;
\end{itemize}
\end{lemma}
In other words, the program $\PiSucc_\P \cup \D_\B$ generates the linear order in $\PTIME$ in the size of data and instances, and in $\EXPTIME$ in the size of process.

\paragraph*{Encoding Stability into Datalog with Negation}

In the following we define a Datalog program with negation 
that, based on the linear order from above, 
produces all maximal extended databases.
Each maximal extended database is going to be encoded 
as one of the SMs of the program.
The program adapts \emph{guess and check} 
methodology from answer-set programming that 
organizes rules in guessing rules that generate SM candidates, 
and checking rules that discards bad candidates.
%

\textbf{Notation}
To encode guessing of executions 
we introduce relations $\Moved$ and $\NotMoved$ of size $g+1$
such that $\Moved(\tk,o)$ means that instance $o$ traverses at step $\tk$, 
and  $\NotMoved(\tk,o)$ means the opposite.
Here, $\NotMoved$ is needed for technical reasons.
Similarly, for transitions we introduce relations $\Trans$ and $\NotTrans$ such that
$\Trans(\tk,t)$ means that at step $\tk$ transition $t$ is traversed; and
$\NotTrans(\tk,t)$ means the opposite. 
Further,
we introduce relation $\Completed$  of size $g$
that we use to keep track of the steps that are completed.
That is, $\Completed(\tk)$ is true if step $\tk$ is completed
and steps that precede $\tk$ are also completed.
Then, to store the positions of each instance after 
some execution step we use relation $\Place$. 
E.g., if after $\tk$-th step instance $o$ is at place $p$ 
then $\Place(\tk,o,p)$ is true.
%
To store facts that are produced up to a certain step 
we introduce prime version relation $R'$ for each $R$ in $\Sigma_{\B}$.
Then $R'(\tk,\ts)$ is true iff $R(\ts)$ is produced up to step $\tk$. 

First we define \emph{guessing rules}:
\begin{align*}
 \InGuess(\tK;O) & \la \Key(\tK), \Inz(O,\blank), \neg \NonInGuess(\tK;O), \\
 \NonInGuess(\tK;O) & \la \Key(\tK), \Inz(O,\blank), \neg \InGuess(\tK;O), \\
 \bot & \la \Key(\tK), \InGuess(\tK;O_1), \InGuess(\tK;O_2), O_1 \neq O_2, \\
 \bot & \la \Inz(O,\blank), \Key(\tK), \neg \Moved(\tK,O).
\end{align*}
Intuitively, the first two
rules enforce each SM to partition instances into 
$\Moved$ and $\NotMoved$ for each step $\tk$,
and the last two ensures that at most one and at least one instance is selected.
%

%
We define the same kind of rules 
for $\Trans$ and $\NotTrans$.

Once an instance and a transition have been selected 
for one execution step $\tk$,
we need to ensure that the instance can actually traverse the transition.
Relation $\Completed$ keeps track of that for each step $\tk$ by checking 
if
\ei the selected instance $o$ satisfies the execution condition of the selected transition $t$;
\eii the instance $o$ is at place $q$ from which $t$ originates; and
\eiii if all previous execution steps were already completed.
This is achieved using the \emph{checking rules}: 
\begin{align*}
	\Done(\tK_2)  \la\ & \InGuess(\tK_2;O), \TGuess(\tK_2;t), \Succ(\tK_1,\tK_2),\\ 
	& \Done(\tK_1), E_t(\tK_1;O), \Place(\tK_1;O;q).
\end{align*}
Condition $E_t(\tK_1;O)$ is similar to the positive acyclic case where the execution $\tomega$ is
replaced with the execution step $\tK_1$.
%
%

Similarly, we define generation rules that for $R'$ and rules that update 
position of instances store in $\Place$.
%
%
%
Let the above rules 
together with the program $\PiSucc_\P \cup \D_\B$
define the program $\Picl$ for closed \DABP{s},
and let $\D_\I$ be a database instance that contains 
$\Inz$ facts and $\D$. 
%

\begin{lemma}
Let $\tk$ be an execution step in $\B$, and let 
$R(\ts)$ be a fact.
The following is equivalent:
\begin{itemize}
    \item There is an execution of length $\tk$ in $\B$ that produces $R(\ts)$;
    \item $\Picl \cup \D_\I  \mdl R'(\tk;\ts)$.
\end{itemize}
\end{lemma}
%

Now we want to test query $Q$ for stability.
We collect new query answers with the rule:
$
Q'(\tX) \la R'_1(\tK;\tu_1),\dots , R'_n(\tK;\tu_n).
$
Let $\PiQtest$ be the test program containing $Q$, $Q'$ and the test rule as in the previous case. Then the following holds:

\[
   \text{ $Q$ is instable in $\PID$ 
    \ iff\ \ }
	\Picl \cup \D_\I \cup \PiQtest \mdl \instable.
	 \]




\subsection{Proof of Proposition~\ref{proposition:cy-cl-process-hd:extended}}
\label{sec:hardness-SMS}

In the following we prove Proposition~\ref{proposition:cy-cl-process-hd:extended} defined above.

In particular, we show how to encode
the brave reasoning under Stable Model Semantics (SMS)
for a given
Datalog program $\Pi \cup \D$
with negation 
into stability problem for 
normal cyclic singleton \DABP $\tup{\P_{\Pi,A},\I_0,\D}$
under closed semantics,
where program $\Pi$ is encoded in the process model 
and data part of the program $\D$ is encoded in the database of the process.
As usual, the test query is $\Qtest$.

\paragraph*{Standard notation for Datalog program with negation}
For Datalog programs under stable model semantics (SMS) 
we use the following notation.
A normal Datalog rule is a rule of the form
		\begin{align*}
			R(\tu) \la R_1(\tu_1),\dots,R_l(\tu_l), \neg R_{l+1}(\tu_{l+1}),\dots, \neg R_h(\tu_h).
		\end{align*}
We use $H$ to denote the head of the rule $R(\tu)$, and 
$A_1,\dots,A_l, \neg A_{l+1},\dots, \neg A_{h}$ 
to denote body atoms 
$R_1(\tu_1),\dots,R_l(\tu_l), \neg R_{l+1}(\tu_{l+1}),\dots, \neg R_h(\tu_h)$
Then we can write the rules $r$ as: 
\[
H  \la A_1,\dots,A_l,\neg A_{l+1},\dots, \neg A_{h}.
\]

We represent a fact $R(\tu)$ as a Datalog \emph{fact rule} $R(\tu) \la$.

A Datalog program with negation $\Pi$ is a finite 
		set of normal Datalog rules $\{r_1,\dots,r_k \}$.

\paragraph*{Grounding of a Datalog program}
		Let $r$ be a normal Datalog rule and $C$ a set of constants.
		The \emph{grounding} $\grn_C(r)$ of $r$ is a set of 
		rules without variables obtained by substituting the variables in $r$ with constants from $C$
		in all possible ways.
		In this way we can obtain several grounded rules from a non-grounded rule.
		The grounding $\grn(\Pi)$ for a program $\Pi$ is a program obtained by grounding rules 
		in $\Pi$ using the constants from $\Pi$.
		We note that program $\grn(\Pi)$ and $\Pi$ have the same semantic properties (they have
		the same SM, see later).
		Program $\grn(\Pi)$ is just an expanded version of $\Pi$ (it can be exponentially bigger than $\Pi$).

\paragraph*{Stable model semantics} 
Concerning stable model semantics we use
the following notation.
An \emph{interpretation} of a program represented as a set of facts.
Let $M$ be an interpretation. 
We define the \emph{reduct} of $\Pi$ for $M$ as the ground 
		positive program 
		\begin{align*}
			\Pi^M = \{ A\la A_1,\dots,A_l \mid\  &
			A \la A_1,\dots,A_l,\neg A_{l+1},\dots, \neg A_h \in \grn(\Pi), \\
												&
					 M \cap \set{A_{l+1},\dots, A_h} = \emptyset \} 
		\end{align*}
		 
		Since $\Pi^M$ is a positive ground program it has a unique Minimal Model (MM), in the inclusion sense	
		Then, 
\[
M \text{ is a \emph{stable model} (SM) of } \Pi  
\quad\text{iff}\quad 
M \text{ is the minimal model of } \Pi^M.
\]
Given a program $\Pi$ and a fact $A$ we say that
		\begin{align*}
			\Pi \mdl A
		\end{align*}
if there exists a SM $M$ of $\Pi$ such that $A \in M$.

For a given $\Pi$ and a fact $A$, deciding whether $\Pi \mdl A$ is $\NEXPTIME$-hard.

\newcommand{\BPDA}{\B_{\P,\D,A}}
\newcommand{\PiD}{\Pi \cup \D}

\paragraph*{Encoding of Brave Entailment into Stability Problem}

Given a program $\Pi \cup \D$ and a fact $A$ we construct an 
\DABP $\B_{\P,\D,A}\tup{\P_{\Pi,A},\I_0,\D}$ such that
for a test query $\Qtest \la \dummy$ the following holds:
\[
	\text{$\Pi \cup \D \mdl A$ iff $\Qtest$ is stable in $\tup{\P_{\Pi,A},\I_0,\D}$}.
\]
For convenience, in the following we use $\Pi$ to denote $\PiD$, unless 
otherwise is stated.

Intuitively, process $\BPDA$ is constructed such that the following holds. 
\begin{itemize}
\item 
The process generates all possible interpretations for 
$\Pi$ using the variables and constants from $\Pi$. 
That is, it generates all possible candidates
for SMs of $\Pi$.
\item 
For every such SM candidate $M$, 
the process checks if $M$ is a SM of $\Pi$ 
by:
\begin{itemize}
	\item[\emph{i)}] computing the MM of $\Pi^M$ denoted with $M'$;
	\item[\emph{ii)}] checking if $M' = M$.
\end{itemize}
\item
	If $M$ is a SM of $\Pi$ then the process checks for the given fact $A$ whether it holds that $A \in M$.
	If so, the process produces $\dummy$.
\end{itemize}

We organize $\B_{\Pi,A}$ in 6 subprocesses represented in Figure~\ref{fig:steps}.

\begin{figure}[!ht]
\begin{center}
\includegraphics[scale=0.35]{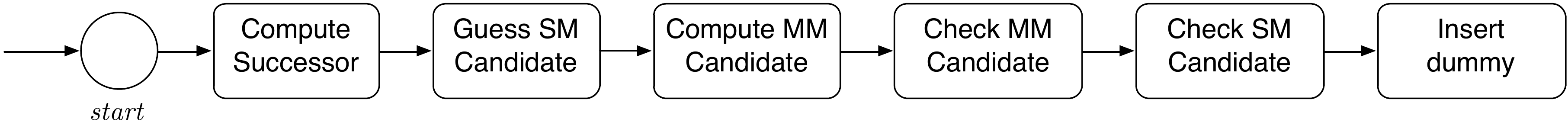}
\end{center}
\caption{Subprocesses composing the process net of ${\P_{\Pi,A}}$.}
\label{fig:steps}
\end{figure}

\bigskip
\noindent
The subprocesses are intuitively defined as follows:
\begin{description}
\item[Subp 1.] 
	\emph{(Compute successor relations)} 
	First we compute the successor relations
	$\Succ^i$ of sufficient size $i$, 
	that we need in the next steps.
	This we need for technical reasons.
\item[Subp 2.] 
	\emph{(Guess a SM candidate)} 
	At this step, the process produces a SM candidate
	by non-deterministically producing facts obtained 
	from relations and constants that appear in the program.
	Let $R$ be a relation in $\Pi$. 	
	Then, for each $R$-fact that can be obtained by taking 	
	the constants from $\Pi$, 
	a process does an execution step at 
	\emph{the choice place} 
	from which if an instance traverse one way
	the process produces this $R$-facts, 
	and if it traverses the other way
	then it does not.
	We denote with $M$ the guessed SM candidate.
\item[Subp 3.] 
	\emph{(Compute a MM candidate of the reduct)} 
	We want compute the MM of the reduct $\Pi^M$.
	To do so, we first compute a candidate $M'$ 
	for the MM by by non-deterministically applying the rules of 
	$\Pi^M$.
	Computing a candidate and the testing if the candidate 
	is the MM is our approach to find the MM. 
\item[Subp 4.] 
	\emph{(Check if $M'$ is the MM of the reduct)}
	At this step we check if $M'$ is indeed the minimal model
	of the reduct $\Pi^M$.
	If this is not the case, the process is not going to progress further.
\item[Subp 5.] 
	\emph{(Check if SM candidate is a SM)}
	If we are at this step then $M'$ is the MM of $\Pi^M$. 
	Now we check if $M' = M$. 
	If this is the case, then $M$ is a stable model of $\Pi$.
\item[Subp 6.] 
	\emph{(Insert $\dummy$)} 
	Finally, we check if $A \in M$. If this is the case then the process produces $\dummy$.
\end{description}

\paragraph*{Instance and data part.}
	We initialize the instance part $\I_0$
	by placing a single instance at the start place, 
	we set database to be the data part of the program $\D$.

\paragraph*{Process model.}
	In the following we construct the process model $\P_{\Pi,A}$.

\paragraph*{Subp 1: Computing successor relations.} 

In order to nondeterministically select which $R$-facts 
to produce for a relation $R$ in $\Pi$,
we introduce sufficiently big linear order that index all $R$-facts.
Since there are exponentially many $R$-facts we define the process
rules that compute the order starting from an order of a
polynomial size.
The rules that compute the exponentially big order
uses the same rules define as in Lemma~\ref{lemma:succ:order}.
Here, the difference is that we use constants from $\Pi$
as digits.

Let $C = \{b_1,\dots,b_c\}$ be the constants from $\Pi$. 
We define a linear order $<$ on $C$ such that
\[b_1 < b_2 < \dots <b_c.\]

Let $<^j$ be the lexicographical order  linear order on $C^j$, 
defined from $<$
for some $j>0$.

Further, let $n$ be the maximum between
\begin{itemize}
	\item the maximal arity of a relation in $\Pi$; and
	\item the largest number of variables in a rule in $\Pi$.
\end{itemize}
We want to compute the successor relation $\Succ^j$ that contains immediate
successors in the order $<^j$ for $j=1,\dots,n$

\textbf{Vocabulary and Symbols}
To encode the order as database relations we
introduce relations:
$\Const$ of size $1$ to store constants from $\Pi$;
$\Succ^j$ of size $2j$ to store 
immediate successors in the order $<^j$;
$\First^j$ and $\Last^j$ to store
the first and the last element of the order $<^j$.
That is,
\begin{itemize}
	\item 
	$\Const(b)$ -- is true iff $b$ is a constant from $\Pi$.
	\item 
	$\Succ^j(\tb,\tb')$ -- is true iff 
	$\tb$ is the immediate successor of $\tb'$
	in the order $<^j$;
	\item 
	$\First^j(\tb)$ -- is true iff 
	$\tb$ is the first element in the order $<^j$;
	\item 
	$\Last^j(\tb)$ --  is true 
	iff $\tb$ is the last element in the order $<^j$.
\end{itemize}

\textbf{Initialization} 
We initialize relations for the ordering as follows:
\begin{itemize}
	\item 
	$\Const(b)$ -- we intialize relation $\Const$ 
	with all the constants from $\Pi$;
	\item 
	$\Succ^1(b,b')$ -- we initialize relation $\Succ^1$ 
	saying that $b'$ is the successor of $b$;
	\item 
	$\First^j(b,\dots,b)$ -- is the initialization for relation $\First^j$ such that $b$ is the first element in  the order $<$;
	\item 
	$\Last^j(b,\dots,b)$ -- is the initialization for relation $\Last^j$ 
	such that $b$ is the last element in the order $<$.
\end{itemize}

\begin{figure}[!ht]
\begin{center}
\includegraphics[scale=0.4]{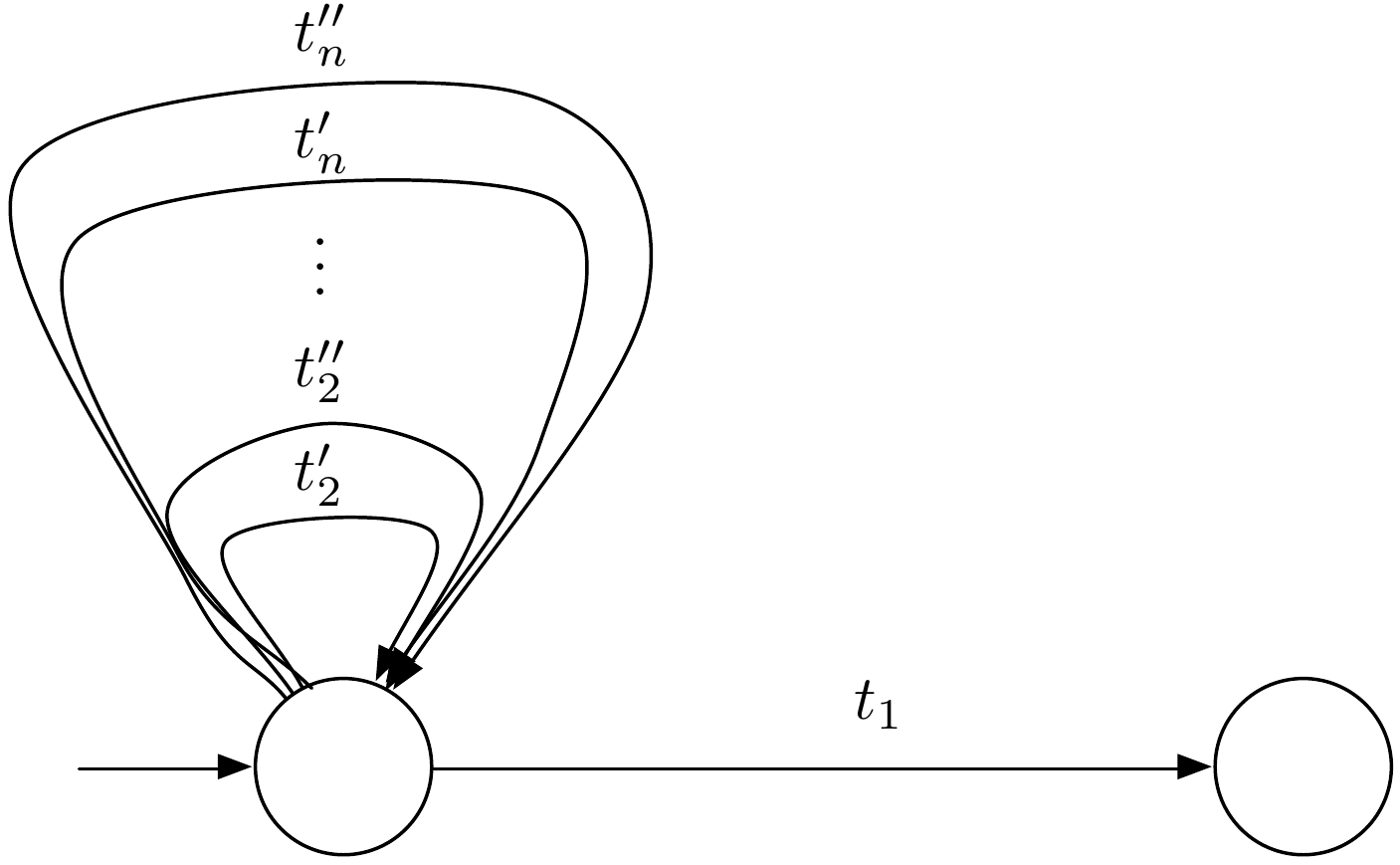}
\end{center}
\caption{Subprocess 1 computes successor relations $\Succ^i$ for $i = 1,\dots,n$}
\label{fig:subp1}
\end{figure}

\textbf{Encoding into the process}
We introduce $2n-1$ transitions 
$t_1,t'_2,t''_2,\dots,t'_n,t''_n$ (see Figure~\ref{fig:subp1})
such that 	$t'_j$ and $t''_j$ are used to generate $\Succ^j$. 
Then, we set the execution condition for these transitions to be always executable:
\begin{align*}
	E_{t'_j} = E_{t''_j} = \true.
\end{align*}

We use the writing rules to populate the relations 
$\Succ^j$ for $ 1 < j \leq n$ 
as follows:

\begin{align*}
	 & W_{t'_{j+1}}: \Succ^{j+1}(Z,\tX,Z,\tY) \la 
	 				 \Const(Z), \Succ^j(\tX,\tY);\\
	 & W_{t''_{j+1}}: \Succ^{j+1}(Z_1,\tX,Z_2,\tY) \la  
	 				\Succ^1(Z_1,Z_2), \First^j(\tX), \Last^j(\tY).
\end{align*}

Once all successor relations are generated transition $t_1$ can be executed:
\begin{align*}
	E_{t_1}: \Succ^n(\tX,\blank), Last^n(\tX).
\end{align*}

\paragraph*{Subp 2: Guessing a SM candidate.} 
Let $R_1,\dots,R_m$ be the relations in $\Pi$.
For every relation $R$ in $\Pi$ we create a subprocess $\textit{Guess-}R$
that non-deterministically guesses $R$-facts that belong to a SM candidate $M$.

Subprocess 2 is composed by connecting subprocess $\textit{Guess-}R$ 
for each relation $R$ as depicted in Figure~\ref{fig:subp2}.

\begin{figure}[!ht]
\begin{center}
\includegraphics[scale=0.4]{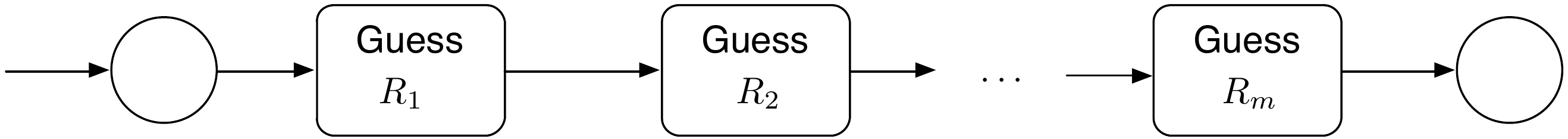}
\end{center}
\caption{Subprocess 2}
\label{fig:subp2}
\end{figure}

\textbf{Notation}
We assume the following notations:
$a$ arity of a relation $R$ in $\Pi$;
$m$ is the number of relations in $\Pi$;
$\done_R$ is a relation of arity $a$ such that 
$\done_R(\tu)$ is true after the subprocess $\textit{Guess-}R$ 
has guessed whether to include $R(\tu)$-fact in the SM candidate or not

\textbf{Encoding into the process}
		The subprocess $\textit{Guess-}R$  is defined as in Figure~\ref{fig:subp2-detail}

\begin{figure}[!ht]
\begin{center}
\includegraphics[scale=0.4]{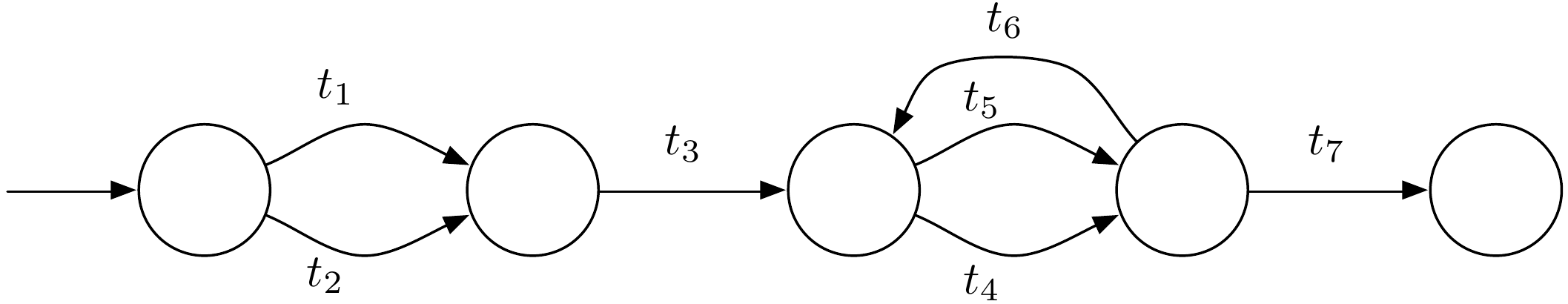}
\end{center}
\caption{Subprocess $\textit{Guess-}R$ }
\label{fig:subp2-detail}
\end{figure}

For convenience introduce condition  
$\CurrentR(\tX)$ that is true if 
the next $R(\tX)$-fact for which the process has to decide
whether to include it in the SM candidate or not.
The condition is defined with:
\begin{align*}
	\CurrentR(\tX): \Succ^a(\tX,\tY), \done_R(\tY), \neg \done_R(\tX).
\end{align*}

Transitions $t_1$ and $t_2$ are executed non-deterministically. Intuitively, they non-deterministically
decide whether the $R$-fact, obtained by grounding $R$ with constants from $\First^a$, 
belongs to the SM candidate ($t_1$) or not ($t_2$):
\begin{align*}
	& E_{t_1} = E_{t_2} : \true; \\
	& W_{t_1}: R(\tX) \la \First^a(\tX); \\
	& W_{t_2}: \true \la \true.
\end{align*}

Then, transition $t_3$ inserts that the guess for the first $R$-fact has been made by inserting $\done_R(\tx)$:
\begin{align*}
	& E_{t_3} : \true; \\
	& W_{t_3}: \done_R(\tX) \la \First^a(\tX).
\end{align*}

Transitions $t_4$ and $t_5$, similarly to transitions $t_1$ and $t_2$, non-deterministically guess whether 
the next $R(\tX)$-fact belongs to the SM candidate or not:
\begin{align*}
	& E_{t_4} = E_{t_5} : \true; \\ 
	& W_{t_4}: R(\tX) \la \CurrentR(\tX); \\
	& W_{t_5}: \true \la \true.
\end{align*}

Transition $t_6$, similarly to transition $t_3$, inserts fact $\done_R(\tX)$ after decision for $R(\tX)$-fact has been made:
\begin{align*}
	& E_{t_6} : \true; \\
	& W_{t_6}: \done_R(\tX) \la \CurrentR(\tX).
\end{align*}

When all guesses have been made, transition $t_7$ can be executed and 
the next subprocess will be executed:
\begin{align*}
	& E_{t_7} : \done_R(\tX), \Last^a(\tX); \\
	& W_{t_7}: \true \la \true.
\end{align*}

\paragraph*{Subp 3: Compute the minimal model of the reduct.} 

The subprocesses 3 and 4 compute the MM $M'$ of $\Pi^M$.
Intuitively, this done in the following way:
\begin{itemize}
	\item 
		Since $\Pi^M$ is a positive ground program the MM of $\Pi^M$
		is unique and it can be computed as the Least Fixed Point (LFP) on the rules in $\Pi^M$.
	\item 
		In subprocess 3, depicted in Figure~\ref{fig:subp3}, the process produces facts 
		that are in the LFP of $\Pi^M$.
		For every relation $R$ we introduce a relation $R'$ that stores facts produced by the 
		LFP 	computation.
	\item 
		In principle, subprocess 3 can produce all facts from the LFP if it executes
		a sufficient number of times. However, it can produce also only a part of the LFP 
		if it decides to traverse $t_{k+1}$.
	\item
		In other words, subprocess 3 non-deterministically decides how many facts from the LFP to produce.
	\item 
		In subprocess 4 we check if all facts from the LFP of $\Pi^M$ are indeed produced
		at subprocess 3.
\end{itemize}

\textbf{Vocabulary and Symbols}
\begin{itemize}
	\item $R'(\tu)$ -- holds iff $R(\tu)$ is in the LFP of $\Pi^M$ (i.e. it is in the MM
		of $\Pi^M$) and it is computed by subprocess 3.
\end{itemize}

\begin{figure}[!ht]
\begin{center}
\includegraphics[scale=0.4]{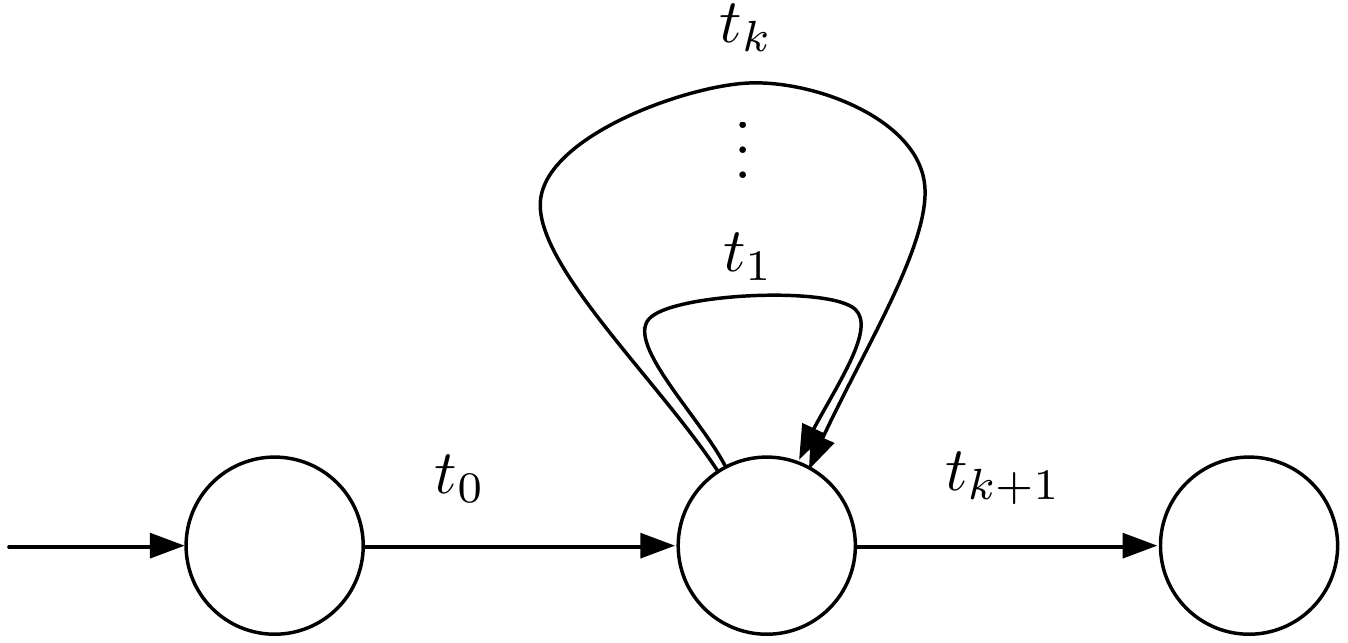}
\end{center}
\caption{Subprocess 3 computes the MM candidate of the reduct}
\label{fig:subp3}
\end{figure}

\textbf{Encoding into the process}
Let $\set{r_1,\dots,r_k}$ be the rules in $\Pi$.
For every rule $r_i$ of the form $H \la A_1,\dots,A_l ,\neg A_{l+1},\dots, \neg A_h$ we introduce transition $t_i$ as
depicted in Figure~\ref{fig:subp3} with execution condition:
\begin{align*}
	E_{t_i}: \true
\end{align*}
%
and writing rule as follows:
\begin{align*}
	W_{t_i}:  H' \la \ & A_1',\dots,A_l',   A_1,\dots,A_l ,\neg A_{l+1},\dots, \neg A_h.
\end{align*}




Here, atoms $H', A_1',\dots,A_l'$ are the same as 
$H, A_1,\dots,A_l$,
except that each relation name $R$ is renamed with $R'$.
Atoms $A_1,\dots,A_l ,\neg A_{l+1},\dots, \neg A_h$ evaluates over $M$ and they are true iff there exists a grounding substitution $\theta$ (a substitution that replaces variables with constants) such that the ground rule 
$\theta A \la \theta  A_1,\dots,\theta A_l$ is in the reduct $\Pi^M$.
For $l=0$, the fact $\theta H'$ is produced by the process since the rule $\theta H \la $ is in $\Pi^M$ as thus $H$ is in the LFP of $\Pi^M$.
For $l>0$, assume that $\theta  A_1',\dots,\theta A_l'$ are already produced by the process such that $\theta  A_1,\dots,\theta A_l$ 
are in the LFP of $\Pi^M$.
Then we have that $\theta H'$ is produced by the process iff 
$\theta H$ is in the LFP of $\Pi^M$.

\paragraph*{Subp 4: Check if the computed model is a minimal model of the reduct.} 
		After the execution of subprocess 3 we obtain a MM candidate $M'$ as a set of
		$R'$-facts produced by the process.
	
		In this step we check if $M'$ is indeed a MM of $\Pi^M$, because in the preceding step
		it may be that the process has generated a part of the LFP of $\Pi^M$.
	
		For the check we define the process as in Figure~\ref{fig:subp4}, where each transition
		$t_{r_i}$ checks if $M'$ contains all facts in the LFP that
		can be produced by the rule $r_i$.

\begin{figure}[!ht]
\begin{center}
\includegraphics[scale=0.4]{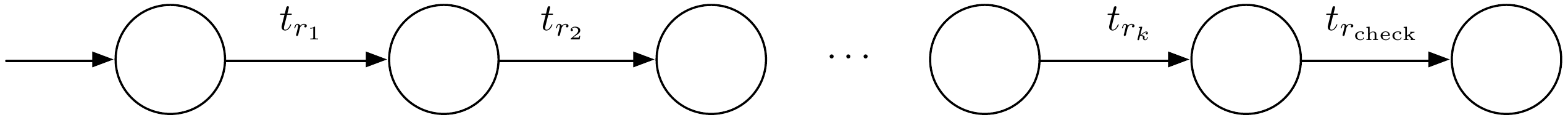}
\end{center}
\caption{Subprocess 4 checks if the MM candidate is the MM of the reduct}
\label{fig:subp4}
\end{figure}

\textbf{Notation}

We introduce unary predicate $\fail$ that is true iff $M'$ is not a MM.

\textbf{Encoding into the process}
	For every rule $r$ we introduce a transition $t_r$ with execution condition
		\begin{align*}
			E_{t_r}: \true,
		\end{align*}
		and with writing rule as follows:
		\begin{align*}
		W_{t_r}: \fail \la & A_1',\dots,A_l', \neg H', A_1,\dots,A_l ,\neg A_{l+1},\dots, \neg A_h
		\end{align*} 	
 	
Fact $\fail$ is produced by the process iff facts $\theta A_1',\dots,\theta A_l'$
are produced by the subprocess 3 while $\theta H'$ is not, for some substitution $\theta$. 
Obviously, this is true iff $M'$ is not the MM of the reduct.

	Last transition $t_{\textit{check}}$ is executable if none of the previous steps
		has generated the $\fail$ predicate:
		\begin{align*}
			E_{t_{\textit{check}}}: \neg \fail.
		\end{align*}

\paragraph*{Subp 5: Checking if SM candidate is a SM.} 
		If the process execution can reach subprocess 5 it means that $M'$ is
		indeed the MM of reduct $\Pi^M$.
		It remains to check if $M$ is a SM of $\Pi$, that is if $M'= M$.

		For this check we define the subprocess as in Figure~\ref{fig:subp5}.

		Transition $t'_i$ checks if there is a $R'_i$-fact for which there is
		no $R_i$-fact and transition $t''_i$ checks if there is a $R_i$-fact for which there is
		no $R'_i$-fact.

\begin{figure}[!ht]
\begin{center}
\includegraphics[scale=0.4]{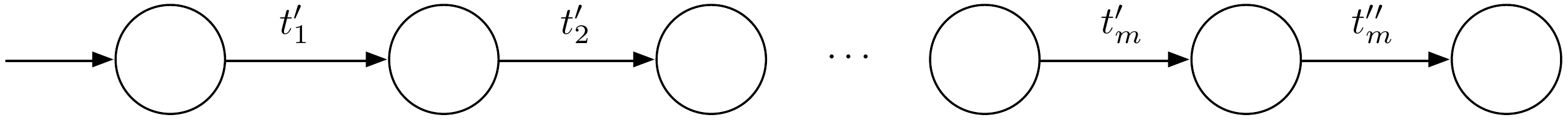}
\end{center}
\caption{Subprocess 5 checks if SM candidate is a SM}
\label{fig:subp5}
\end{figure}

\textbf{Notation}
We introduce unary predicate $\faili$ that holds if $M' \neq M$.

\textbf{Encoding into the process}
		Transition $t'_i$ is encoded as follows:
		\begin{align*}
			W_{t'_i}: \faili \la R'_i(X), \neg R_i(X).
		\end{align*}

		Transition $t''_i$ is encoded as follows:
		\begin{align*}
			W_{t''_i}: \faili \la R_i(X), \neg R'_i(X).
		\end{align*}

\paragraph*{Subp 6: Insert $\dummy$.} 
		After the execution of subprocess 5 if no $\faili$ facts were produced,
			then $M' = M$
	
		Subprocess 6 checks whether this is the case. If $M' = M$ and $A \in M$ the process inserts $\dummy$

\begin{figure}[!ht]
\begin{center}
\includegraphics[scale=0.4]{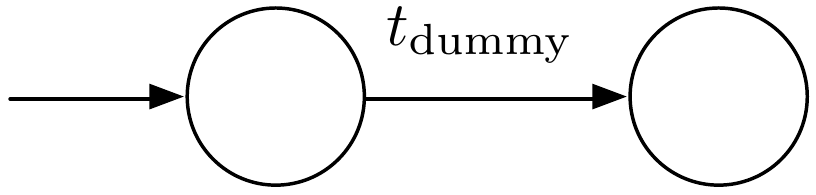}
\end{center}
\caption{Subprocess 6 inserts $\dummy$}
\label{fig:subp6}
\end{figure}

\textbf{Encoding into the process}
The subprocess is depicted in Figure~\ref{fig:subp6}
	Transition $t_{\dummy}$ checks if $M'= M$ with the execution condition:
	\begin{align*}
		E_{t_{\dummy}}: \neg \textit{fail}_1,\dots,\neg \textit{fail}_m.
	\end{align*}
	
		By traversing $t_{\dummy}$ if condition $A \in M$ then $\dummy$ is inserted with the
		following writing rule:
		\begin{align*}
			W_{t_{\dummy}}: \dummy \la A.
		\end{align*}

All together, we have that fact $A$ is produced by the process iff
there exists a SM of the program that contains $A$.
This concludes the proof.



\newpage
\section{Rowo \DABP{s}}

\subsection{Proof of Theorem~\ref{theorem:rowo}}

We first show encodings for fresh rowo \DABP{s}.
Then we show encodings for closed rowo  \DABP{s}, and 
finally we combine  these two encodings to 
obtain encodings for arbitrary rowo \DABP{s}.

\paragraph*{Rowo Fresh \DABP{s}}
\label{subsec:rowo:fresh:open}

First we analyze a fresh rowo $\B=\tup{\P,\D}$.
For this case
we adapt $\PiPposfr$
from Section~\ref{sec:DABP-Datalog-power-short}.

In rowo  \DABP{s},  we have that each instance needs not to traverse 
a transition more than once in order to produce the most that
the transition can produce.
It may need to traverse some transitions more than once 
to reach other transitions, but in total it is sufficient that
it makes at most $m^2$ traversals to reach all transitions,
where $m$ is the number of transitions.
I.e., it is sufficient to consider executions of a single instance of
maximal size $m^2$.
Therefore, 
we can eliminate recursion from traversal rules in positive fresh \DABP{s} by
creating a bounded derivation of maximal size $m^2$.

To this end, instead of $\In_p$ we introduce relations
$\In^i_p$ 
for each $i$ up to $m^2$
to record that a fresh instance can reach place $p$ in $i$ steps. 
That is,
$\In^i_p(\ts)$ is true iff a fresh instance with $\In(\ts)$-record
can reach place $p$ in $i$ steps.

Next we adapt the rules.

\myparagraph{Traversal Rules}
For each transition $t$ from a place $q$ to a place $p$
and for each $i$ up to $m^2$ we adapt a \emph{traversal rule} as follows:
\[
    \In^{i+1}_p(\ts) \la E_t^i(\ts)
\]
Note that, as a difference from the general case, here $E_t^i(\ts)$ denotes the execution  condition evaluated over the initial database 
(rather than on the extended database
as $E'_t(\ts)$ would denote), and where
$\In(\ts)$ condition is replaced with $\In^{i}_q(\ts)$.
The database relations are not changed.

\myparagraph{Generation Rules}
Similarly, for each transition $t$ from above we introduce the following 
\emph{generation rule}:
\[
		  R'(\tu) \la E_t^i(\ts) ,B_t^i(\ts',\tu), \ts=\ts'
\]
As pointed in the observations 
negation in $E_t$ and $B_t$ 
does not make reasoning more complex 
{since negation is on the base relations 
that are not updated by the process}. 

\myparagraph{Summary}
Let $\PiProwofr$ be the non-recursive Datalog program with stratified negation
that encodes the rowo process $\P$
obtained from $\PiPposfr$
substituting the traversal and generation rules with the rules above.
The rest of the program is the same as in the positive variant.

\begin{lemma}
Let $R'(\tu)$ be a fact defined over $\adomExt$, then the following is equivalent:
\begin{itemize}
  \item there is an execution in $\B$ that produces $R(\tu)$
  \item $\PiProwofr \cup \D \models R'(\tu)$
\end{itemize}
\end{lemma}

Let $\PiQtest$ be the test
program based on a query $Q$ as defined for the general variant.

\begin{theorem} 
\label{rowo-th-instability}
The following is equivalent:
\begin{itemize}
	\item $Q$ is instable in $\B$;
	\item $\PiProwofr \cup \D \cup \PiQtest \models \instable$
\end{itemize}
\end{theorem}


\paragraph*{Rowo Closed \DABP{s}}
\label{subsec:rowo:normal:arbitrary:closed}

We now consider a possibly cyclic rowo $\B=\tup{\P,\I,\D}$ under closed executions.
We adapt the encoding from the  acyclic closed variant
(which can be obtained from closed positive).
The main difference is that each instance is encoded 
independently of the others.
I.e., we encode an execution of a single instance as a tuple 
$\tomega$ of the form
\[
	\tomega=\tup{o; t_{h_1}, \dots, t_{h_i}}.
\]
{meaning that instance $o$ traverses first $t_{h_1}$ then $t_{h_2}$, and so on.}

{Similarly we adapt $R^i$'s and $\State^i$'s from the general case
such that:}
\begin{itemize}
	\item 
	$R^i(o;t_{h_1}, \dots, t_{h_i};\ts)$ 
		denotes that the instance $o$ after traversing
		$t_{h_1}, \dots, t_{h_i}$ produces $R(\ts)$; and
	\item 
	$\State^i(o;t_1,\dots,t_i;p)$  
		denotes that the instance $o$ 
		after traversing $t_1,\dots,t_i$  
		is located at place $p$.
\end{itemize}
Similarly to the previous variant, cycles can be dealt with
bounded derivations of maximal length $m^2$,
so $i$ ranges from $1,\dots,m^2$.
%
Then similarly to the closed variant, we use
$\In^0(o;\ts)$ to associate instance $o$ with the input record $\In(\ts)$.
In this way,  we obtain the facts that can be produced by each instance.
Then we introduce additional rules that combine facts 
produced by different instances.
Assume we are given a query 
$Q(\tX) \la R_1(\tu_1),\dots,R_n(\tu_n)$
that we want to check for stability.
To this end, we introduce the following relations.

\begin{itemize}
	\item 
	$\Path$  
	is a relation with arity $m^2+1$ that contains 
	legal paths of an instance. 
	$\Path(o;\ttr,\tepsilon)$ is true iff
	$\ttr$ is a legal path in $\P$ 
	for instance $o$.
	For technical reasons we introduce $\epsilon$ to denote an empty transition.
	Then, $\tepsilon$ is vector of $\epsilon$
	that we use to fill in remaining positions in $\Path$
	($|\tepsilon| = m-|\ttr|$).
	\item 
	$R'$ is an auxiliary relation of size $1+m^2+\arity(R)$
	 that we introduce for each 
	$R$ in $\B$ to store $R$-facts produced by an instance.
	That is, $R'(o;\ttr,\tepsilon;\ts)$ in true iff
	$R(\ts)$ is produced after $o$ traversed $\ttr$.
	\item 
	$\Exec^j$ 
	are relations of arity  $(m^2 \times j) + j$ 
	for every $j = 1,\dots,k$ 
	that combines legal paths for different $n$ instances 
	where $n$ is is the number of atoms in the query.
	Then, $\Exec^j(o_1,\ttr_1,\dots,o_j,\ttr_j)$ is true iff
	tuple $\ttr_l$ is a legal path for instance $o_l$ 
	and if $o_h = o_l$ then $\ttr_h = \ttr_l$.
	This relation we use to record all combinations of instances 
	that can contribute to create a new query answer,
	thus if two facts are produced by the same instance  
	($o_h = o_l$) then the facts have to be produced
	on the same legal path ($\ttr_h = \ttr_l$).
\end{itemize}
Again $i$ ranges from $1,\dots,m^2$.
Now we define a program that generates those relations.
%


\myparagraph{Initialization Rules}
First we adapt the \emph{initialization rules} for a single instance:
\begin{align*}
& \State^0(o;q)\la \true 
	\text{ iff $q$ is the starting place of instance $o$, and } \\
& R^0(O;\tY) \la \In'(O;\blank), R(\tY).
\end{align*}

\myparagraph{Traversal Rules}
Similarly, we adapt traversal rules to be for a single instance, as follows:
\begin{align*}
	\State^{i+1}(O;\tT,t;p) \la 
		\State^{i}(O;\tT;q), E_t(O)
\end{align*}
for each transition $t$ from place $q$ to place $p$
and $E_t(O)$ is the same as $E_t$ except that
each atom $\In(\ts)$ is replaced with $\In^0(O,\ts)$
and $\tT$ is a vector of different variables of size $i$.

\myparagraph{Generation Rules}
For a transition $t$ with writing rule $W_t \colon R(\tu) \la B_t(\tu)$,
the generation rules become:
\begin{align*}
	 R^{i+1}(O;\tT,t;\tu) 
		&\la 
	\State^{i+1}(O;\tT,t;p), B_t(O;\tT;\tu).
\end{align*}
Here $B_t(O;\tT;\tu)$ is obtained analogously to $E_t(O)$.

\myparagraph{Copy Rules}
Then we adapt the copy rules as follows:
\begin{align*}
	 R^{i+1}(O;\tT, t ; \tU) &\la R^i(O;\tT; \tU), \text{ and } 
	\\
	 R'(O;\tT,\tepsilon; \tU) &\la R^i(O;\tT; \tU)
\end{align*}
where the size of $\epsilon$-vector $|\tepsilon| = m^2 - i$.

\myparagraph{Summary}
We denote the above program as $\PiProwocl$. 
Then we have that the following holds:

\newcommand{\tplt}{\bar t}

\begin{lemma}
Let $o$ be an instance in $\B$,
a list of transitions $\tplt$ in $\B$ of size $i$,
and $R(\tu)$ a fact, 
then the following is equivalent:

\begin{itemize}
  \item  
  after $o$ traverses $\tplt$ the fact $R(\tu)$ is produced;
  \item 
  $\PiProwocl \cup \D \models R^i(o;\tplt;\tu)$.
\end{itemize}

\end{lemma}


\myparagraph{Combining Rules}
Then we need to combine the atoms produced by the different instances, e.g.,
\[R^i_1(o_1;\tplt_1;\ts_1)\]
  	\[\vdots\]
\[R^i_n(o_n;\tplt_n;\ts_n).\]
and ensure
that atoms produced by one instance are all produced following one path.
To do this we need to ensure that 
if $o_i = o_j  \Rightarrow \tplt_{i} = \tplt_j$.
This is achieved with the \emph{combining rules}.

First we copy all legal 
paths of an instance from $\State^i$ into $\Path$ relation:
\begin{align*}
	\Path(O;\tT,\bar \epsilon) \la \State^i(O;\tT;\blank)
		\text{ where } |\epsilon| = m^2 - i.
\end{align*}

Then we initialize $\Exec^1$ with all legal paths of an instance.
\[
    \Exec^1(O;\tT) \la \Path(O;\tT).
\]
Then we combine different paths in the following way.
If instance $O_l$ executes $\tT_l$ and the same instance executes $\tT_{i+1}$ then 
$\tT_l$ and $\tT_{i+1}$ must be the same.
This is captured with the following rules:
\begin{align*}
			\Exec^{j+1}(O_1,& \tT_1,\dots,O_l,\tT_l,\dots,O_j,\tT_j,O_{j+1},\tT_{j+1}) \la \\
  			& \Exec^j(O_1,\tT_1,\dots,O_l,\tT_l,\dots,O_j,\tT_j),\\
  			& \Path(O_{j+1};\tT_{j+1}),\\
  			& O_l = O_{j+1}, \tT_l = \tT_{j+1}
\end{align*}
for every $l \in 1,\dots,i$.

If the instance $O_{j+1}$ is different from all the other instances 
$O_1,\dots,O_j$ then executing path of $O_{j+1}$ can be any legal path

\begin{align*}
			\Exec^{j+1}(O_1,& \tT_1,\dots,O_l,\tT_l,\dots,O_j,\tT_j,O_{j+1},\tT_{j+1}) \la \\
  			& \Exec^j(O_1,\tT_1,\dots,O_l,\tT_l,\dots,O_j,\tT_j),\\
  			& \Path(O_{j+1};\tT_{j+1}),\\
  			& \neg (O_1 = O_{j+1}), \neg (O_2 = O_{j+1}), \dots, \neg (O_j = O_{j+1}).
\end{align*}

\myparagraph{$Q'$-rule}
Then, $Q'$ collects what has been produced for relations $R_1,\dots,R_n$
for the give query $Q(\tX) \la R_1(\tu_1),\dots,R_n(\tu_n)$ with the rule
\begin{align*}
			Q'(\tX) \la & \Exec^n(O_1,\tW_1,\dots,O_n,\tW_n),\\
			& R'_1(O_1;\tW_1;\tu_1),\\
			& \dots\\
			& R'_n(O_n;\tW_n;\tu_n). 
\end{align*}

\myparagraph{Test Rule}
The \emph{test rule} is then as before:
\[
  	\instable \la  Q'(\tX), \neg Q(\tX).
\]

\myparagraph{Summary}
Let us denote with $\PiPIQtestrowo$ 
the testing program for $Q$ defined above.
The program is non-recursive Datalog with stratified negation.

Let $\D_\I$ be the database that encodes the instance part $\In$
and that contains database $\D$.
\begin{theorem}
		\label{th-rowo-closed}
The following are equivalent:
\begin{itemize}
	\item $Q$ is instable in $\B$ under closed executions;
	\item $\PiProwocl \cup \D_\I \cup \PiPIQtestrowo \models \instable$.
\end{itemize}
\end{theorem}
%


\paragraph{Rowo \DABP{s}}
\label{subsec:rowo:normal:arbitrary:open}
{
For arbitrary rowo, similarly to the positive
variants, the encoding is obtained combining the encodings for the fresh
and the closed variants.
 
To combine what comes from the instances in the process and the new ones it is enough to add rules that will combine 
the program for closed ($\PiPQrowocl$) and 
fresh($\PiPQrowofr$).

To this end, for a given query 
$Q(\tX) \la R_1(\tu_1),\dots,R_n(\tu_n)$, 
we introduce  relations: 
\begin{itemize}
  \item 
  $B^i_Q$ of arity 
  $\arity(R_1)+\dots+\arity(R_n)$ 
  that contains on the first $i$ arguments what comes from a mixture of
  existing and new process instances while the others come only from existing process instances, for $i=1,\dots,n$.
\end{itemize}

\paragraph{Encoding into Non-Recursive Datalog}

Let $\B =\tup{\P,\I,\D}$ be a rowo  \DABP.

Now we define rules that compute relations introduced above.

First we consider what is produced by the running instances 
\begin{align*}
	B^0_Q(\tY_1,\dots,\tY_n) \la & \Exec^n(O_1,\tW_1,\dots,O_n,\tW_n),\\
			& R'_1(O_1;\tW_1;\tu_1),\\
			& \dots\\
			& R'_n(O_n;\tW_n;\tu_n). 
\end{align*}
	
Then, for the $i$-th atom we both consider the case in which it was produced by a new instance (1)
and the case it was produced by the instances already in the process (2). These cases are added to the combinations
obtained for the atoms from $1$ to $i-1$. We do this for every $i = 1,\dots,n$.
\begin{align}
	B^{i}_Q(\dots,& \tY_{i-1},\tY_{i},\tY_{i+1},\dots) \la\\ \nonumber
		& B^{i-1}_Q(\dots,\tY_{i-1},\blank,\tY_{i+1},\dots), R'_i(\tY_i) \\
	B^{i}_Q(\dots,& \tY_{i-1},\tY_{i},\tY_{i+1},\dots) \la\\ \nonumber
		& B^{i-1}_Q(\dots,\tY_{i-1},\tY_{i},\tY_{i+1},\dots).
\end{align}

Then, we add the $Q'$-\emph{rule} to collect what has been produced
by the process for relations $R_1(\tu_1),\dots,R_n(\tu_n)$ as follows:
\begin{align*}
Q'(\tX) \la B^n_Q(\tY_1, \dots, \tY_n).
\end{align*}

The above rules extend the testing program 
$\PiPIQtestrowo$  for normal cyclic arbitrary closed.
We denote the new testing program with 
$\PiPIQtestrowoOpen$.
\begin{itemize}
	\item $Q$ is instable in $\B$;
	\item $\PiProwocl \cup \PiProwofr \cup
		 \D_\I  \cup \PiPIQtestrowoOpen \models \instable$.
\end{itemize}

Then we set $\PiPQrowo = \PiProwocl \cup \PiProwofr \cup \PiPIQtestrowoOpen$,
and the claim follows from there.

\subsection{Proof of Proposition~\ref{prop-rowo:generability:complexity}}
\begin{proof}
   Assume the instance $o$ is at place $p$ 
   and it has an input record $I(\bar s)=M_S(o)$.
   To show the claim it is sufficient 
   to guess a closed execution $\Upsilon$ 
   consisting of the traversals by $o$,  
   and then verify whether atoms $A_1,\ldots,A_n$ can
   be produced by such execution.
   In the case of singleton rowo \DABP{s} under closed executions,
   a closed execution is uniquely determined by 
   a path in $\P$. 
   Thus, we guess a path $t_1,\ldots,t_m$ in $\P$  
   that starts in $p$.
   This guess is polynomial in the size of $\P$.
   For all transitions on the path, we further guess
    assignments 
   $\alpha_1,\ldots,\alpha_m$ for the execution conditions
   $E_{t_1},\ldots,E_{t_m}$. 
   Then for each writing rule $W_{t_i}$ of the transition $t_i$
   we guess up to $n$ assignments 
   $\beta^l_{t_i}$ for $1 \le l \le n$, 
   because a rule may need to produce more than 
   one fact, but no more than the size of the query.
   In principle, only a subset of
   $W_{t_1},$ \ldots, $W_{t_m}$ may be
   needed to produce atoms  $A_1,\ldots,A_n$.
   Wlog we can guess the assignments for all.
   Now we verify. 
   Firstly, we verify whether
   the path can be traversed by the instance.
   This is, if for every execution condition $E_{t_i}$
   the ground query $\alpha_i E_{t_i}$  evaluates to 
   $\true$  in $\D\cup\set{\In(\bar s)}$.
   Secondly, we verify  whether $A_1,\ldots,A_n$
   are produced on the path by checking if 
   for every $A_i$ there exist a writing rule
   $W_{t_j} \col A_{t_j} \la B_{t_j}$ 
   and assignment  $\beta^l_{t_j}$ such that 
   the ground query $\beta^l_{t_j} B_{t_j}$ 
   evaluates to $\true$ in $\D\cup\set{\In(\bar s)}$ and 
   $A_i$ is equal with the head of the writing rule
    $\beta^l_{t_j} A_{t_j}$.
   Since all guesses and checks are polynomial
   in the size of $\B$ the claim
   follows directly.
\end{proof}

\end{document}